\documentclass[a4paper,11pt]{article}

\usepackage[margin=2.4cm]{geometry}
\usepackage{amsthm}

\usepackage{dsfont}     
\usepackage{amssymb}     
\usepackage{bbm}        

\usepackage{mathtools}

\usepackage{algorithm}
\usepackage{algorithmic}

\usepackage{amsmath}

\usepackage{enumerate}

\usepackage{paralist}

\usepackage{amsfonts}
\usepackage{amscd}

\usepackage{xcolor}

\usepackage{tikz}
\usetikzlibrary{arrows,automata} 

\tikzstyle{rect} = [rectangle, rounded corners, minimum width=1cm, minimum height=1cm,text centered, draw=black ]
\tikzstyle{arrow} = [thick,->,>=stealth]

\usepackage{subcaption}

\usepackage{hyperref}

\newcommand{\EE}{\mathbb{E}}
\newcommand{\ZZ}{\mathbb{Z}}
\newcommand{\PP}{\mathbb{P}}
\newcommand{\RR}{\mathbb{R}}
\newcommand{\NN}{\mathbb{N}}
\newcommand{\1}{\mathds{1}}
\newcommand{\m}{\mathbb}
\newcommand{\mc}{\mathcal}
\newcommand{\eps}{\varepsilon}
\newcommand{\defas}{\coloneqq}

\newcommand{\until}{\, .. \,}

\DeclareMathOperator{\Supp}{supp}
\DeclareMathOperator{\argmin}{argmin}

\usepackage[normalem]{ulem} 

\theoremstyle{plain} 
\newtheorem{@theorem}{Theorem}[section]
\newenvironment{theorem}{\begin{@theorem}}{\end{@theorem}}
\newtheorem{Lemma}[@theorem]{Lemma}
\newtheorem{Corollary}[@theorem]{Corollary}

\newtheorem{Proposition}[@theorem]{Proposition}

\theoremstyle{definition} 
\newtheorem{Definition}[@theorem]{Definition}
\newtheorem{Example}[@theorem]{Example}

\theoremstyle{remark} 
\newtheorem{Remark}[@theorem]{Remark}

\usepackage{multirow} 
\usepackage{hhline} 

\usepackage{thmtools, thm-restate}
\theoremstyle{definition}

\begin{document}

\title{ 
\Large Finite-memory strategies in POMDPs \\ with Long-run Average Objectives 
\thanks{Partially supported by  Austrian Science Fund (FWF) NFN Grant No RiSE/SHiNE S11407, by CONICYT-Chile through grant PII 20150140, by
    ECOS-CONICYT through grant C15E03, and by COST Action GAMENET. This project also benefited from the support of the FMJH Program PGMO RSG 2018-0031H and from the support of EDF, Thales, Orange and Criteo. }  }
\author{ 
Krishnendu Chatterjee \thanks{IST Austria, Klosterneuburg, Austria.} \\
\and
Raimundo Saona $^\dag$ \\ 
\and Bruno Ziliotto \thanks{CEREMADE, CNRS, Universit\'e Paris Dauphine, PSL Research Institute, Paris, France.}
}

\date{}

\maketitle

\begin{abstract}
Partially Observable Markov Decision Processes (POMDPs) is a standard model for dynamic systems with probabilistic and nondeterministic behaviour in uncertain environments. We prove that in POMDPs with long-run average objective, the decision-maker has approximately optimal strategies with finite memory. This implies notably that approximating the long-run value is recursively enumerable, as well as a characterization of the continuity property of the value with respect to the transition function.   
\end{abstract}

\thispagestyle{empty}
\clearpage
\setcounter{page}{1}

\section{Introduction} 
\label{sec: Introduction}
In a Partially Observable Markov Decision Process (POMDP), at each stage, the decision-maker chooses an action that determines, together with the current state, a stage reward and the distribution over the next state. The state dynamic is imperfectly observed by the decision-maker, who receives a stage signal on the current state before playing. Thus, POMDPs generalize the Markov Decision Process (MDP) model of Bellman \cite{Bellman_57}. 

POMDPs are widely used in prominent applications such as in computational 
biology~\cite{Bio-Book}, software verification~\cite{CCHRS11}, reinforcement learning~\cite{LearningSurvey}, to name a few. 
Even special cases of POMDPs, namely, probabilistic automata or 
blind MDPs, where there is only one signal, is also a standard model
in several applications~\cite{Rabin63, PazBook, Bukharaev}. 

In many of these applications, the duration of the problem is huge. Thus, considerable attention has been devoted to the study of POMDPs with long duration. A standard way is to consider the long-run objective criterion, where the total reward is the expectation of the inferior limit average reward (see \cite{SURVEYAC} for a survey). The value for this problem is known to coincide with several classical definitions of long-run values (asymptotic value, uniform value, general uniform value, long-run average value, uncertain-duration process value \cite{RSV02, R11, RV16, NS10, VZ16}) and has been characterized in \cite{RV16}. In this paper, we will simply name this common object \textit{value}. Thus, strong results are available concerning the existence and characterization of the value. 

This is in sharp contrast with the study of long-run optimal strategies. Indeed, before our work, little was known about the sophistication of strategies that approximate the value. It has been shown that:(i) stationary strategies approximate the value in MDPs~\cite{B62}; and (ii) belief-stationary strategies approximate the value in blind MDPs~\cite{RSV02} and POMDPs with an ergodic structure~\cite{Borkar_2000}. 
\\

Our main contributions are:
\begin{itemize}
    \item {\em Strategy complexity.}
    We show that for every POMDP with long-run average objectives, for every $\eps>0$, there is a 
    finite-memory strategy (i.e. generated by a finite state automaton) that achieves expected reward within $\eps$ of the optimal value. 
    In the case of blind MDP finite memory is equivalent to finite recall (i.e. decisions are defined using only the last actions), but finite recall cannot achieve $\eps$-approximations in general POMDPs.

    \item {\em Computational complexity.}
    An important consequence of our above result is that the decision version of the approximation problem for POMDPs with long-run average objectives (see Definition \ref{Definition: Decision version of approximating the value}) is {\em recursively enumerable (r.e.)} but not decidable. 
    Our results on strategy complexity imply the recursively enumerable upper bound and the lower bound is a consequence of~\cite{MHC03}.
    
    \item {\em Value property.}
    The long-run reward of a finite-memory strategy is robust upon small perturbations of the transition function, where the notion of perturbation over the transition function is defined as in Solan \cite{Solan03} and Solan and Vieille \cite{SV10}. This implies lower semi-continuity of the value function upon such small perturbations. This result is tight in the sense that there is an example with a discontinuous value function (see Example~\ref{Example: Discontinuity}).
\end{itemize}
A natural question would be to ask for an upper bound on the size of the memory needed to generate $\eps$-optimal strategies, in terms of the data of the POMDP. In fact, a previous \textit{undecidability} result \cite{MHC03} shows that such an upper bound can not exist (see Subsection \ref{Sec: Complexity}). Thus, the existence of $\eps$-optimal strategies with finite memory is, in some sense, the best possible result one can have in terms of strategy complexity. 

\section{Model and statement of results}
\label{sec: Preliminaries}
\subsection{Model}
\label{sec: Model}

Throughout the paper we mostly use the following notation: (i) sets are denoted by calligraphic letters, e.g. $\mc{A}, \mc{H}, \mc{K}, \mc{S}$; (ii) elements of these sets are denoted by lowercase letters, e.g. $a, h, k, s$; and (iii) random elements with values in these sets are denoted by uppercase letters, e.g. $A, H, K, S$. For a set $\mc{C}$, denote $\Delta(\mc{C})$ the set of probability measure distributions over $\mc{C}$, and $\delta_c$ the Dirac measure at some element $c \in \mathcal{C}$. We will slightly abuse notation by not making a distinction between a probability measure (which can be evaluated on events) and its corresponding probability density (which can be evaluated on elements). 

Consider a POMDP $\Gamma = (\mc{K}, \mc{A}, \mc{S}, q, g)$, with finite state space $\mc{K}$, finite action set $\mc{A}$, finite signal set $\mc{S}$, transition function $q \colon \mc{K} \times \mc{A} \rightarrow \Delta(\mc{K} \times \mc{S})$ and reward function $g \colon \mc{K} \times \mc{A} \rightarrow [0,1]$.

Given $p_1 \in \Delta(\mc{K})$, called \emph{initial belief}, the POMDP starting from $p_1$ is denoted by $\Gamma(p_1)$ and proceeds as follows: 
\begin{itemize}
\item
An initial state $K_1$ is drawn from $p_1$. The decision-maker knows $p_1$ but does not know $K_1$.
\item
At each stage $m \geq 1$, the decision-maker takes some action $A_m \in \mc{A}$. This action determines a stage reward 
$G_m \defas g(K_m, A_m)$, where $K_m$ is the (random) state at stage $m$. Then, the pair $(K_{m+1}, S_m)$ is drawn from $q(K_m, A_m)$. The next state is $K_{m+1}$ and the decision-maker is informed of the signal $S_m$, but neither of the reward $G_m$ nor of the state $K_{m+1}$. 
\end{itemize}
At stage $m$, the decision-maker remembers all the past actions and signals, which is called $\textit{history before stage $m$}$. Let $\mc{H}_m \defas (\mc{A} \times \mc{S})^{m-1}$ be the set of histories before stage $m$, with the convenient notation $(\mc{A} \times \mc{S})^{0} \defas \{\emptyset\}$. A strategy is a mapping $\sigma \colon \cup_{m \geq 1} \mc{H}_m \to \mc{A}$. The set of strategies is denoted by $\Sigma$. The randomness introduced by the transition function, $q \colon \mc{K} \times \mc{A} \rightarrow \Delta(\mc{K} \times \mc{S})$, suggests that a history $h_m \in \mc{H}_m$ can occur under many sequences of states $(k_1, k_2, \ldots, k_{m-1})$. The infinite sequence $(k_1, a_1, s_1, k_2, a_2, s_2, \ldots)$ is called a \emph{play}, and the set of all plays is denoted by $\Omega$. 

For $p_1 \in \Delta(\mc{K})$ and $\sigma \in \Sigma$, define $\m{P}^{p_1}_{\sigma}$ the law induced by $\sigma$ and the initial belief $p_1$ on the set of plays of the game $\Omega = (\mc{K} \times \mc{A} \times \mc{S})^{\NN}$,
and $\m{E}^{p_1}_{\sigma}$ the expectation with respect to this law. 
For simplicity, identify $\mc{K}$ with the set of extremal points of $\Delta(\mc{K})$.

Let 
\begin{equation*}
    \gamma^{p_1}_{\infty}(\sigma)
        \defas  \m{E}^{p_1}_{\sigma} \left(\liminf_{n \rightarrow +\infty} \frac{1}{n} \sum_{m=1}^n G_m \right) \,,
\end{equation*}
and 
\begin{equation*}
    v_{\infty}(p_1)
        \defas \sup_{\sigma \in \Sigma} \gamma_{\infty}^{p_1}(\sigma) \,.
\end{equation*}
The term $\gamma_{\infty}^{p_1}(\sigma)$ is the long-term reward given by strategy $\sigma$ and $v_{\infty}(p_1)$ is the optimal long-term reward, called \textit{value}, defined as the supremum long-term reward over all strategies.

\begin{Remark}
It has been shown that $v_\infty$ 
coincides with the limit of the value of the $n$-stage problem and $\lambda$-discounted problem, as well as the uniform value and weighted uniform value (see \cite{RSV02,R11,RV16,VZ16}). In particular, we have:
\begin{align*}
    v_\infty(p_1)
        &= \lim_{n \rightarrow +\infty} \sup_{\sigma \in \Sigma} \m{E}^{p_1}_{\sigma} \left(\frac{1}{n} \sum_{m=1}^n G_m \right) \\
        &= \lim_{\lambda \rightarrow 0} \sup_{\sigma \in \Sigma} \m{E}^{p_1}_{\sigma} \left(\sum_{m \geq 1} \lambda(1-\lambda)^{m - 1} G_m \right) \\
        &= \sup_{\sigma \in \Sigma} \liminf_{n \rightarrow +\infty} \m{E}^{p_1}_{\sigma} \left( \frac{1}{n} \sum_{m=1}^n G_m \right) \,.
\end{align*}
\end{Remark}

\begin{Remark}
In the literature, the concept of strategy that we defined is often called \emph{pure strategy}, by contrast with \emph{behavior strategies} that use randomness by allowing strategies of the form $\sigma \colon \cup_{m \geq 1} \mathcal{H}_m \to \Delta(\mc{A})$. By Kuhn's theorem, enlarging the set of pure strategies to behaviour strategies does not change $v_{\infty}$ (see \cite{VZ16, F96}), and thus does not change our results. 
\end{Remark}

\begin{Definition}[Blind MDP]
A POMDP is called \emph{blind MDP} if the signal set is a singleton. 
\end{Definition}

Note that in a blind MDP, signals do not convey any relevant information. Therefore, a strategy is simply an infinite sequence of actions $(a_1,a_2,\dots) \in \mc{A}^{\m{N}}$.

\subsection{Contribution}
We start by defining several classes of strategies. Recall that $\Gamma(p_1)$ is the POMDP $\Gamma$ starting from $p_1$, which is known to the player.
\begin{Definition}[$\eps$-optimal strategy]
Let $p_1 \in \Delta(\mc{K})$ and $\eps>0$. A strategy $\sigma \in \Sigma$ is \emph{$\eps$-optimal} in $\Gamma(p_1)$ if 
\[
    \gamma^{p_1}_{\infty}(\sigma) \geq v_{\infty}(p_1) - \eps \,.
    \]

\end{Definition}
\begin{Definition}[finite-memory strategy]
    A strategy $\sigma$ is said to have \emph{finite memory} if it can be modeled by a finite-state transducer. Formally, $\sigma = (\sigma_u, \sigma_a, \mc{M}, m_0)$, where $\mc{M}$ is a finite set of memory states, $m_0$ is the initial memory state, $\sigma_a: \mc{M} \to \mc{A}$ is the action selection function and $\sigma_u \colon \mc{M} \times \mc{A} \times \mc{S} \to \mc{M}$ is the memory update function. 
\end{Definition}
\begin{Definition}[Finite-recall strategy]
\label{Def: Finite recall strategy}
    A strategy $\sigma$ is said to have \emph{finite recall} if there exists a constant $M > 0$ such that for all $h_M \in \mc{H}_M$, and for all $m > M$ and $h_{m - M} \in \mc{H}_{m - M}$, we have that $\sigma(h_{m - M}, h_M)$ does not depend on $h_{m-M}$. 
\end{Definition}

\begin{Remark}
For blind MDPs, finite-recall and finite-memory strategies coincide with the set of \textit{eventually periodic strategies}: a strategy $\sigma=(a_1,a_2,\ldots)$ is eventually periodic if there exists $T \geq 1$ and $N \geq 1$ such that for all $m \geq N$, $a_{m+T}=a_m$. This property does not extend to general POMDPs (see Proposition \ref{prop: example}): any finite-recall strategy has finite-memory, but the inverse is not true.
\end{Remark}
\begin{Remark}
Finite-memory strategies and finite-recall strategies have been investigated in the Shapley zero-sum stochastic game model \cite{Sha53}. In this framework, none of these strategies is enough to approximate the value, and a long-standing open problem is whether finite-memory strategies \textit{with a clock} are good enough (see \cite{HIN18, HIN20} for more details on this topic). 
\end{Remark}

Our main result is the following theorem.
\begin{theorem} 
\label{theo: POMDP}
For every POMDP $\Gamma$, initial belief $p_1$ and $\eps > 0$, there exists an $\eps$-optimal finite-memory strategy in $\Gamma(p_1)$. 
\end{theorem}
\begin{Remark}
A previous complexity result \cite{MHC03} shows that the size of the memory can not be bounded from above in terms of the data of the POMDP (see Subsection \ref{Sec: Complexity}). 
\end{Remark}
\begin{Corollary} 
\label{theo: blind MDP}
For every blind MDP $\Gamma$, initial belief $p_1$ and $\eps > 0$, there exists an $\eps$-optimal finite-memory strategy in $\Gamma(p_1)$, and thus the strategy is eventually periodic and has finite recall. 
\end{Corollary}

Lastly, finite-recall is not enough to ensure $\eps$-optimality in general POMDPs. 
\begin{Proposition} \label{prop: example}
There exists a POMDP, and $\eps > 0$, with no $\eps$-optimal finite-recall strategy.  
\end{Proposition}

The rest of the paper is organized as follows. Section \ref{Sec: Consequences} explains the consequences of our result in terms of complexity and model robustness. Section \ref{Sec: Examples} introduces examples used to prove negative results and to illustrate our techniques. Section \ref{Sec: Structure of the proof} introduces two key lemmata, and shows that they imply Theorem \ref{theo: POMDP}. Section \ref{sec: Technique: Super-support} proves one of the two lemmata and develops what we call \textit{super-support} based strategies in details. Missing proofs are in the appendices.

\section{Consequences of the results}
\label{Sec: Consequences}

\subsection{Complexity}
\label{Sec: Complexity}

\noindent{\em Decidability.}
A decision problem consists in deciding between two options given an input (accepting or rejecting) and its complexity is characterized by Turing machines. 
A Turing machine takes an input and, if it halts, it either accepts or rejects it. 
If it halts for all possible inputs in a finite number of steps, then the Turing machine is considered an algorithm. 
An algorithm solves a decision problem if it takes the correct decision for all inputs. 
The class of decision problems that are solvable by an algorithm is called {\em decidable}.
Two natural generalizations of decidable problems are: {\em recursively enumerable (r.e.)} and {\em co-recursively enumerable (co-r.e.)}.
The decision problems in r.e. (resp., co-r.e.) are those for which there is a Turing machine that accepts (resp., rejects) every input that should be accepted (resp., rejected) according to the problem, but, on other inputs, it needs not to halt.

Notice that the class of decidable problems is the intersection of r.e. and co-r.e.
In this work, the algorithmic problem of interest is the following.

\begin{Definition}[Decision version of approximating the value]
\label{Definition: Decision version of approximating the value}
    Let $p_1 \in \Delta(\mc{K})$. Given $x \in [0, 1]$, $\eps > 0$ such that $v_\infty(p_1) > x + \eps$ or $v_\infty(p_1) < x - \eps$, the problem consists in deciding which one is the case: to accept means to prove that $v_\infty(p_1) > x + \eps$ holds, while to reject means to prove the opposite.
\end{Definition}


\smallskip\noindent{\em Previous results and implication of our result.}
It is known that the decision version of the approximation problem is not decidable~\cite{MHC03} (even for blind MDPs). 
However, the complexity characterization has been open. Thanks to Theorem \ref{theo: POMDP}, we can design a Turing machine that accepts every input that should be accepted for this problem. 

Consider playing a finite-memory strategy $\sigma$. Then, the dynamics of the game can be described by a finite Markov chain. Therefore, the reward obtained by playing $\sigma$ (i.e. $\gamma^{p_1}_\infty(\sigma)$) can be deduced from its stationary measure, which  can be computed in polynomial time by solving a linear programming problem~\cite[Section 2.9, page 70]{FV97}. Our protocol checks the reward given by every finite-memory strategy to approximate the value of the game $v_\infty(p_1)$. By Theorem \ref{theo: POMDP}, if $v_\infty(p_1) > x + \eps$ holds, a finite-memory strategy that achieves a reward strictly greater than $(x + \eps)$ will be eventually found and our protocol will accept the input. 
On the other hand, if $v_\infty(p_1) < x - \eps$, the protocol will never find out that this is the case because there are infinitely many finite-memory strategies, so it will not halt. 
Thus, our result establishes that the approximation version of the problem is in r.e., 
and the previously known results imply that the problem is not decidable. Formally, we have the following result.
\begin{Corollary} 
\label{Corollary: r.e.-completeness of decision problem}
The decision version of approximating the value is r.e. but not decidable.
\end{Corollary}
\begin{Remark}
The former paragraph shows that no upper bound on the size of the memory used by $\eps$-optimal strategies can be proved. Indeed, if such a bound existed, one could modify the previous algorithm in the following way: reject the input if every finite-memory strategy of size lower than the bound has been enumerated. This would imply that the decision version of approximating the value is decidable, which is a contradiction. 
\end{Remark}
\subsection{Objective comparison}

In this section, we contrast our results with other natural objectives.

Recall that the value of $\Gamma(p_1)$ is defined as
    \[
    v_\infty(p_1) = \sup_{\sigma \in \Sigma} \EE_{\sigma}^{p_1}\left( \liminf_{n \rightarrow \infty} \frac{1}{n} \sum_{m = 1}^{n} G_m \right) \,.
    \]
We say this is a \emph{liminf-average} objective. Consider replacing $\liminf_{n \rightarrow \infty} \frac{1}{n} \sum_{m = 1}^{n} G_m$ by: 
(i) $\limsup_{n \rightarrow \infty} \frac{1}{n} \sum_{m = 1}^{n} G_m$, which we call \emph{limsup-average} objective; (ii) $\limsup_{n \rightarrow \infty} G_n$, which we call \emph{limsup} objective.

\begin{Proposition}
\label{Proposition: Example of no finite memory for different objectives}
    For both limsup-average and limsup objective, there exists a POMDP, and $\eps > 0$, with no $\eps$-optimal finite-memory strategy.  
\end{Proposition}

This negative result, proved in Section \ref{Apendix: Proof of Proposition, no finite memory}, does not imply any computational complexity characterization for the limsup-average or limsup objective, and whether the approximation of the value problem for limsup-average objectives is recursively enumerable remains open.
However, it shows that any approach based on finite-memory strategies cannot establish recursively enumerable bounds for the approximation problem. 

Let us focus on the limsup objective. Limsup objective is arguably simpler than the liminf-average objective and, to formalize this statement, we can compare the complexity of the objects themselves irrespective of any particular context or model (such as POMDPs). The \emph{Borel hierarchy} describes the complexity of an objective by the number of quantifier alternations needed to describe it. Its construction is similar to that of the Borel $\sigma$-algebra, or $\sigma$-field, and is defined as follows.

\begin{Definition}[Borel hierarchy]
    Consider $h_m \in \mc{H}_m = (\mc{A} \times \mc{S})^{m - 1}$ a finite history of the game. The cylinder set generated by $h_m$ is given by $\{h_m\} \times (\mc{A} \times \mc{S})^\NN$. Finite intersection, unions and complements of the cylinder sets generated by finite histories form the first level in the hierarchy. Countable unions of the first level form $\Sigma_1$ and countable intersections form $\Pi_1$. The next level is always obtained from the previous one: countable unions of $\Pi_i$ give $\Sigma_{i+1}$ and countable intersections of $\Sigma_i$ give $\Pi_{i+1}$. The nested sequence of family of problems $\{\Sigma_i \cup \Pi_i\}_{i \geq 1}$ is called \emph{Borel hierarchy}.
\end{Definition}

For example, limsup objective can be described as countable intersection of countable unions of rewards: given a family of sets $(\mc{C}_n)_{n \geq 1}$, $\limsup_{n \rightarrow \infty} \mc{C}_n = \cap_{n \geq 1} \cup_{m \geq n} \mc{C}_m$. The formal result is the following (see~\cite{Ch07}).

\begin{Proposition}
     The limsup objective is $\Pi_2$-complete, i.e. complete for the second level of the Borel hierarchy, whereas the liminf-average objective is $\Pi_3$-complete, i.e. complete for the third level of the Borel hierarchy.
\end{Proposition}

While the notion of Borel hierarchy characterizes the topological complexity for objectives, 
a similar notion of \emph{Arithmetic hierarchy} characterizes the computational complexity for decision problems.


\begin{Definition}[Arithmetic hierarchy]
    Denote $\Sigma_0^1$ the class of r.e. problems and $\Pi_0^1$ the co-r.e. problems. For $i > 1$, define $\Sigma_0^i$ as the class of problems solved by Turing machines with access to oracles for $\Pi_0^{i-1}$ and $\Pi_0^i$ is similarly defined with oracles for $\Sigma_0^{i-1}$. The nested sequence of family of problems $\{\Sigma_0^i \cup \Pi_0^i\}_{i \geq 1}$ is called \emph{Arithmetic hierarchy}.
\end{Definition}

By Corollary \ref{Corollary: r.e.-completeness of decision problem}, we have that POMDPs with a liminf-average objective is in $\Sigma_0^1 \setminus (\Sigma_0^1 \cap \Pi_0^1)$. On the other hand, it was shown in~\cite{BGB12, BG09} that POMDPs with limsup objective with boolean rewards is $\Sigma_0^2$-complete.


We conclude this section with a summary chart contrasting liminf-average and limsup objectives. The surprising result is the complexity switch: limsup objective has lower Borel hierarchy complexity but higher Arithmetic hierarchy complexity in the context of POMDPs. 

\begin{figure}[H]
\centering
    \begin{tabular}{ |p{3cm}||p{4.5cm}|p{4.5cm}|  }
    \hline
    \multicolumn{3}{|c|}
    {Objective comparison in POMDPs} \\
    \hline
    Objective & Borel hierarchy & Arithmetic Hierarchy \\
    \hhline{|=||=|=|}
    limsup & $\Pi_2$-complete & $\Sigma_0^2$-complete \\
    \hline
    liminf-average & $\Pi_3$-complete & $\Sigma_0^1 \setminus (\Sigma_0^1 \cap \Pi_0^1)$ \\
    \hline
    \end{tabular}
    
    \caption{Objective comparison in POMDPs}
\end{figure}

\subsection{Robust $\eps$-optimal strategies}

Consider a POMDP $\Gamma=(\mc{K}, \mc{A}, \mc{S}, q, g)$. It is well known that the value function is continuous with respect to perturbations of the reward function $g$ and the initial belief $p_1$. Now, we show a robustness result concerning the transition function $q$.


In applications, just as in any stochastic model, the structure of the model is decided first, and then the specific probabilities are either estimated or fixed. The values of transition probabilities are approximations: an $\eps$-perturbation of these probabilities are not expected to have an impact on the modelling. In our setting, the transitions are encoded in the function $q \colon \mc{K} \times \mc{A} \to \Delta(\mc{K} \times \mc{S})$ and we would expect some robustness against perturbations of the values it takes.


The notion of perturbation over $q$ is measured as in Solan \cite{Solan03} and Solan and Vieille \cite{SV10}, where perturbations in each transition probability are measured as relative differences, not additive differences. Formally, define the semimetric
    \[
    d(q,q')
        =\max_{\substack{k \in \mc{K}, \ a \in \mc{A} \\ k' \in \mc{K}, \ s \in \mc{S}}} \left\{
            \frac{q(k,a)(k',s)}{q'(k,a)(k',s)} \,,
            \frac{q'(k,a)(k',s)}{q(k,a)(k',s)} 
            \right\}
            - 1 \,.
    \]

Under this notion, and taking $q$ and $q'$ close to each other, we can prove the existence of strategies which are approximately optimal for the POMDP corresponding to $q$ and perform almost as well when they are applied to the POMDP corresponding to $q'$. To formally state this notion of robustness, let us give the following definition.


\begin{Definition}[Robust strategies]
    Given a POMDP $\Gamma$ with transition function $q$, an initial belief $p_1 \in \Delta(\mc{K})$, we say that $\sigma$ is a \emph{robust strategy} for $\Gamma(p_1)$ if the following condition holds: $\forall \eta > 0$ $\exists \delta > 0$ such that
        \[
        d(q, q') \leq \delta \quad \Rightarrow \quad \gamma'^{p_1}_\infty(\sigma) \geq \gamma^{p_1}_\infty(\sigma) - \eta \,,
        \]
    where $\gamma_\infty$ is the long-term reward in $\Gamma$ and $\gamma'_{\infty}$ is the long-term reward in $\Gamma'=(\mc{K}, \mc{A}, \mc{S}, g, q')$. 
\end{Definition}

\begin{Lemma}
    Any finite-memory strategy is robust. Thus, in any POMDP and for any $\varepsilon>0$, there exists a robust $\eps$-optimal finite-memory strategy.  
\end{Lemma}

\begin{proof}

Let $\Gamma=(\mc{K}, \mc{A}, \mc{S}, g, q)$ be a POMDP. Consider $\sigma= (\sigma_u, \sigma_a, \mathcal{M}, m_0)$ a finite-memory strategy for $\Gamma(p_1)$. Playing $\sigma$ from $p_1$ induces a Markov Chain $(Y_n)_{n \geq 1}$ on $\mc{K} \times \mc{S} \times \mc{M}$. Define $\tilde{g} \colon \mc{K} \times \mc{M} \to [0, 1]$ by $\tilde{g}(k,m)  \defas  g(k,\sigma_a(m))$.

Now, consider the 0-Player stochastic game with reward $\tilde{g}$ and transitions given by the kernel of the Markov Chain $(Y_n)_{n \geq 1}$. Let $s_0 \in \mc{S}$ be any signal. By definition, for any $k \in \mc{K}$, the value of this stochastic game starting from $(k,s_0,m_0)$ coincides with $\gamma^k_\infty(\sigma)$. Using \cite{Solan03}[Theorem 6, page 841], we deduce that
    \[
    \gamma'^k_{\infty}(\sigma) \geq \gamma^k_{\infty}(\sigma)-4 |\mc{K}| |\mc{S}| |\mc{M}| d(q,q') \,.
    \]
Integrating $k$ over $p_1$ yields
    \[
    \gamma'^{p_1}_{\infty}(\sigma) \geq \gamma^{p_1}_{\infty}(\sigma)-4 |\mc{K}| |\mc{S}| |\mc{M}| d(q,q') \,.
    \]
Taking $\delta = \eta (4 |\mc{K}| |\mc{S}| |\mc{M}|)^{-1}$, we conclude that $\sigma$ is a robust strategy. By Theorem \ref{theo: POMDP}, for all $\eps>0$, there exists an $\eps$-optimal finite-memory strategy, which is thus robust. 
\end{proof}

\begin{Corollary}
    Let $\mc{K}, \mc{A}, \mc{S}$ be finite sets, $g \colon \mc{K} \times \mc{A} \rightarrow \m{R}$ a reward function, and $p_1 \in \Delta(\mc{K})$ an initial belief. The mapping from $(\Delta(\mc{K} \times \mc{S})^{\mc{K} \times \mc{A}}, d)$ to $\m{R}$ that maps each transition function $q$ to the value at $p_1$ of the POMDP $(\mc{K}, \mc{A}, \mc{S}, g, q)$ is lower semi-continuous. 
\end{Corollary}
\begin{proof}
Let $q \in \Delta(\mc{K} \times \mc{S})^{\mc{K} \times \mc{A}}$. Let $\Gamma = (\mc{K}, \mc{A}, \mc{S}, q, g)$. By the previous lemma, for all $\eps > 0$, there exists $\sigma_{\eps}$ a robust $\eps$-optimal strategy in $\Gamma(p_1)$. Take $\eta = \eps$, by robustness of $\sigma_\eps$, there exists $\delta > 0$ such that, for all $q' \in \Delta(\mc{K} \times \mc{S})^{\mc{K} \times \mc{A}}$, we have that if $d(q, q') \leq \delta$, then $\gamma'^{p_1}_{\infty}(\sigma_{\eps}) \geq \gamma^{p_1}_{\infty}(\sigma_{\eps}) - \eps$. Also, by $\eps$-optimality of $\sigma_\eps$, we have that $\gamma^{p_1}_\infty(\sigma_{\eps}) \geq v_\infty(p_1) - \eps$. Then,
    \[
    v'_\infty(p_1)
        \geq \gamma'^{p_1}_{\infty}(\sigma_{\eps}) 
        \geq v_{\infty}(p_1) - 2 \eps \,.
    \]
Taking $\eps \to 0$, we conclude that
    \[
    \liminf_{q' \to q} v'_{\infty}(p_1) 
        \geq v_{\infty}(p_1) \,,
    \]
and thus $v_\infty$ is lower semi-continuous with respect to $q$.
\end{proof}

Lower semi-continuity of the value function is the best result one can achieve in the following sense.

\begin{Proposition}
\label{Proposition: discontinuity}
	There is a POMDP such that the mapping from $(\Delta(\mc{K} \times \mc{S})^{\mc{K} \times \mc{A}}, d)$ to $\m{R}$ that maps each transition function $q$ to the value at $p_1$ of the POMDP $(\mc{K}, \mc{A}, \mc{S}, g, q)$ is discontinuous.
\end{Proposition}


\section{Examples}
\label{Sec: Examples}

In this section, we introduce examples to prove negative results (Propositions \ref{prop: example}, \ref{Proposition: Example of no finite memory for different objectives} and \ref{Proposition: discontinuity}) and to illustrate our techniques later on.

\subsection{Negative results}
\label{Sec: Negative results}

Let us prove Propositions \ref{prop: example} and \ref{Proposition: Example of no finite memory for different objectives} by presenting an example for each statement.

\subsubsection{Proof of Proposition \ref{prop: example}}

We will prove that there exists a POMDP and $\varepsilon>0$ with no $\varepsilon$-optimal finite-recall strategy by an explicit construction. Recall that a strategy has finite recall if it uses only a finite number of the last stages in the current history to decide the next action (see Definition \ref{Def: Finite recall strategy}). Therefore, our construction should have the property that, for any finite-recall strategy, there is a pair of finite histories such that:
\begin{enumerate}
    \item The last stages are identical, i.e., the player did the same actions and received the same signals in the last part of both histories (but the starting point was different).
    
    \item Taking the same decision in both histories leads to losing some reward that can not be compensated in the long-run. 
    
    \item The previous loss does not decrease to zero by increasing the amount of memory.
\end{enumerate}

\begin{Example}
\label{Exam: Finite recall is not enough for POMDPs}
Consider the POMDP $\Gamma = (\mc{K}, \mc{A}, \mc{S}, q, g)$ with five states: $k_0, k_1, \ldots, k_4$. The initial state is $k_0$ and players know it (formally, the initial belief is $\delta_{k_0}$). The state $k_4$ is an absorbing state from where it is impossible to get out and rewards are zero. The states $k_1$ and $k_2$ form a sub-game where the optimal strategy is trivial. This is the same for the state $k_3$. From $k_0$ a random initial signal is given indicating which sub-game the state moved to. The key idea is that there is an arbitrarily long sequence of actions and signals which can be gotten in both sub-games, but the optimal strategy behaves differently in each of them. Therefore, to forget the initial signal of the POMDP leads to at most half of the optimal value.

Figure \ref{Figure: Finite recall for POMDPs} is a representation of $\Gamma$: first under action $a$ and then action $b$. Each state is followed by the corresponding reward, and the arrows include the probability for the corresponding transition along with the signal obtained.

\begin{figure}[h] 
\begin{subfigure}{0.5 \textwidth}
\centering
    \begin{tikzpicture}[->,>=stealth',shorten >=1pt,auto,node distance=3cm, semithick]
    
    \node[state] (k0) [rect] {$k_0 | 0$};
    \node[state] (k1) [rect, below of = k0, xshift = 1.5cm, yshift = 20] {$k_1 | 1$};
    \node[state] (k2) [rect, below of = k1, yshift = 10] {$k_2 | 1$};
    \node[state] (k3) [rect, below of = k0, xshift = -1.5cm, yshift = 20] {$k_3 | 1$};
    \node[state] (k4) [rect, below of = k3, yshift = 10] {$k_4 | 0$};

    \path   (k0) edge [above, sloped] node {$1/2; s_2$} (k1)
    (k1) edge node {$1/2; s_2$} (k2)
    (k2) edge node {$1; s_1$} (k4)
    (k1) edge [loop right] node {$1/2; s_1$} (k1)
    (k0) edge [above, sloped] node {$1/2; s_1$} (k3)
    (k3) edge [loop left] node {$1/2; s_1$} (k3)
    (k3) edge [loop below] node {$1/2; s_2$} (k3)
    (k4) edge [loop left] node {$1; s_1$} (k4);
    
    \end{tikzpicture}
    \caption{Action $a$}
\end{subfigure}
\begin{subfigure}{0.5 \textwidth}
\centering
    \begin{tikzpicture}[->,>=stealth',shorten >=1pt,auto,node distance=3cm, semithick]
    
    \node[state] (k0) [rect] {$k_0 | 0$};
    \node[state] (k1) [rect, below of = k0, xshift = 1.5cm, yshift = 20] {$k_1 | 1$};
    \node[state] (k2) [rect, below of = k1, yshift = 10] {$k_2 | 1$};
    \node[state] (k3) [rect, below of = k0, xshift = -1.5cm, yshift = 20] {$k_3 | 1$};
    \node[state] (k4) [rect, below of = k3, yshift = 10] {$k_4 | 0$};

    \path   (k0) edge [above, sloped] node {$1/2; s_2$} (k1)
    (k1) edge [above, sloped] node {$1; s_1$} (k4)
    (k2) edge [right] node {$1; s_1$} (k1)
    (k0) edge [above, sloped] node {$1/2; s_1$} (k3)
    (k3) edge [left] node {$1; s_1$} (k4)
    (k4) edge [loop left] node {$1; s_1$} (k4);
    
    \end{tikzpicture}
    \caption{Action $b$}
\end{subfigure}
    \caption{Finite recall is not enough for POMDPs}
    \label{Figure: Finite recall for POMDPs}
\end{figure}
\end{Example}

\newpage

The sub-game of $k_1$ and $k_2$ has a unique optimal strategy: play action $a$ until receiving signal $s_2$, then play action $b$ once and repeat. The value of this sub-game is $1$ and deviating from the prescribed strategy would lead to a long-run reward of $0$. Similarly, the value of the sub-game of $k_3$ has a unique optimal strategy: to always play action $a$. Again, the value of this sub-game is $1$ and playing any other strategy leads to a long-run reward of $0$.

By the previous discussion, the value of this game starting from $k_0$ is $1$. On the other hand, the maximum value obtained by strategies with finite recall is only $1/2$, by playing, for example, always action $a$. Finite-recall strategies achieve at most $1/2$ because, no matter how much finite recall there is, by playing the game the decision-maker faces a history of having played always action $a$ and always receiving a signal $s_1$, except for the last signal which is $s_2$. Then, if action $b$ is played, the second sub-game is lost; if action $a$ is played, the first sub-game is lost. That is why, for any $0 < \eps < 1/2$, there is no $\eps$-optimal finite-recall strategy for this POMDP.

\subsubsection{Proof of Proposition \ref{Proposition: Example of no finite memory for different objectives}}
\label{Apendix: Proof of Proposition, no finite memory}

We will show that for the limsup-average and limsup objectives there is a blind MDP where there is no $\eps$-optimal finite-memory strategy. For both cases, the example is constructed with the following idea in mind. To achieve the optimal value, the decision-maker needs to play an action $a_1$ for some period, then play another action $a_2$ and repeat the process. The key is to require that the length of the period gets longer as the game progresses. This kind of strategy can not be achieved with finite-memory strategies. 

For the limsup-average objective, the blind MDP example is due to Venel and Ziliotto~\cite{VZ20} and is presented below.

\begin{Example}
    Consider two states $k_0$ and $k_1$ and the player receives a reward only when the state is $k_1$. To reach $k_1$, the decision-maker can play action \emph{change} and move between the two states. By playing action \emph{wait}, the state does not change. 
\end{Example}
    
Figure \ref{Figure: Finite recall is not enough for POMDPs} is a representation of the game.
    
\begin{figure}[h] 
\begin{subfigure}{0.5\textwidth}
\centering
    \begin{tikzpicture}[->,>=stealth',shorten >=1pt,auto,node distance=3cm, semithick]
    
    \node[state] (k0) [rect] {$k_0 | 0$};
    \node[state] (k1) [rect, right of = k0] {$k_1 | 1$};    
    
    \path   
    (k0) edge[bend left] node {$1$} (k1)
    (k1) edge[bend left] node {$1$} (k0);
    
    \end{tikzpicture}
    \caption{Action \emph{change}}
\end{subfigure}
\begin{subfigure}{0.5\textwidth}
\centering
    \begin{tikzpicture}[->,>=stealth',shorten >=1pt,auto,node distance=3cm, semithick]
    
    \node[state] (k0) [rect] {$k_0 | 0$};
    \node[state] (k1) [rect, right of = k0] {$k_1 | 1$};    
    
    \path   
    (k0) edge[loop above] node {$1$} (k1)
    (k1) edge[loop above] node {$1$} (k0);
    
    \end{tikzpicture}
    \caption{Action \emph{wait}}    
    
\end{subfigure}        
\caption{Finite recall is not enough for POMDPs}
\label{Figure: Finite recall is not enough for POMDPs}
\end{figure}

Consider the initial belief $p_1 = \frac{1}{2} \cdot \delta_{k_0} + \frac{1}{2} \cdot \delta_{k_1}$, the uniform distribution. It is easy to see that finite-memory strategies (or equivalently finite-recall strategies) can not achieve more than $1/2$. On the other hand, the value of this game with the \emph{limsup-average} objective is $1$, and is guaranteed by the following strategy: 
    \[
    \sigma = (wait)^{2^{0^2}} (change) (wait)^{2^{1^2}} \cdots (change)(wait)^{2^{N^2}} \cdots 
    \]
Hence, finite-memory strategies do no guarantee any approximation for POMDPs with \emph{limsup-average} objective.

For the \emph{limsup} objective, the blind MDP example is the following.

\begin{Example}
    Consider four states ($k_0, k_1, k_2$ and $k_3$) and two actions: \emph{wait} ($w$) and \emph{change} ($c$). The initial state is $k_1$ and players know it. In $k_1$, if $w$ is played, then the state moves to $k_2$ with probability $1/2$ and stays with probability $1/2$; if $c$ is played, then the absorbing state $k_0$ is reached. From state $k_2$, if we play $w$, we stay in the same state; if we play $c$, we move to state $k_3$. From $k_3$, the only state that has a positive reward, if we play any action, we return to the initial state $k_1$. 
\end{Example}
    
Figure \ref{Figure: Finite-memory is not enough for limsup objective} is a representation of the game.

\begin{figure}[h] 
\begin{subfigure}{0.5\textwidth}
\centering
    \begin{tikzpicture}[->,>=stealth',shorten >=1pt,auto,node distance=3cm, semithick]
    
    \node[state] (k0) [rect] {$k_0 | 0$};
    \node[state] (k1) [rect, right of = k0] {$k_1 | 0$};    
    \node[state] (k2) [rect, below of = k1] {$k_2 | 0$};
    \node[state] (k3) [rect, left of = k2] {$k_3 | 1$};    
    
    \path   
    (k0) edge [loop left] node {$1$} (k0)
    (k1) edge node {$1$} (k0)
    (k2) edge node {$1$} (k3)
    (k3) edge node {$1$} (k1);
    
    \end{tikzpicture}
    \caption{Action \emph{change}}
\end{subfigure}
\begin{subfigure}{0.5\textwidth}
\centering
    \begin{tikzpicture}[->,>=stealth',shorten >=1pt,auto,node distance=3cm, semithick]
    
    \node[state] (k0) [rect] {$k_0 | 0$};
    \node[state] (k1) [rect, right of = k0] {$k_1 | 0$};    
    \node[state] (k2) [rect, below of = k1] {$k_2 | 0$};
    \node[state] (k3) [rect, left of = k2] {$k_3 | 1$};
    
    \path   
    (k0) edge [loop left] node {$1$} (k0)
    (k1) edge [loop left] node {$1/2$} (k1)
    (k1) edge node {$1/2$} (k2)
    (k2) edge [loop left] node {$1$} (k2)
    (k3) edge node {$1$} (k1);
    
    \end{tikzpicture}
    \caption{Action \emph{wait}}    
\end{subfigure}        
\caption{Finite-memory is not enough for \emph{limsup} objective}
\label{Figure: Finite-memory is not enough for limsup objective}
\end{figure}

In this blind MDP (see \cite{BG09,BGB12}), for the \emph{limsup} objective, for any $\eps > 0$, there is an infinite-memory strategy that guarantees $1 - \eps$, so the value of the game is $1$. On the other hand, applying any finite-memory strategy (or equivalently finite-recall strategies) yields a \textit{limsup }reward of $0$. Hence, finite-memory strategies do no guarantee any approximation for POMDPs with \emph{limsup} objective.

\subsubsection{Proof of Proposition \ref{Proposition: discontinuity}}

We will show that there is a POMDP with discontinuous value with respect to the transition function. The idea is to have two possible scenarios where signals are slightly different. By analyzing a long sequence of signals, the player is able to identify which is the scenario of the current state and so take a better strategy. The following example considers a transition function parameterized by $\eps \ge 0$.

\begin{Example}
\label{Example: Discontinuity}
	Consider two states $(k_u, k_d)$ and three actions: \emph{up} ($a_u$), \emph{down} ($a_d$) and \emph{wait} ($a_w$). Signals are relevant only for action $a_w$: under a non-symmetric transition function, they inform about the underlying state. More concretely, there are two signals $s_u$ and $s_d$. Playing actions $a_u$ or $a_d$ will give signal $s_u$ or $s_d$ respectively, adding no information. Playing action $a_w$ leads to signals $s_u$ and $s_d$ with slightly different probabilities if the state is $k_u$ or $k_d$. In terms of actions and rewards, $a_u$ leads to positive reward only if the state is $k_u$, similarly, $a_d$ leads to positive reward only if the state is $k_d$. Finally, $a_w$ leads to null reward in both states. Figure~\ref{Figure: Discontinuity} is a representation of this POMDP where transitions are specified.
\end{Example}

\begin{figure}[h]
\begin{subfigure}[b]{.34\textwidth}
\centering
    \begin{tikzpicture}[->,>=stealth',shorten >=1pt,auto,node distance=3cm, semithick]
    
    \node[state] (ku) [rect] {$k_u | 0$};
    \node[state] (kd) [rect, below of = ku, yshift = 20] {$k_d | 0$};
    
    \path   
    (ku) edge [loop left] node {$1/2; s_u$} (ku)
    (ku) edge [loop right] node {$1/2; s_d$} (ku)
    (kd) edge [loop left] node {$1/2 - \eps; s_u$} (kd)
    (kd) edge [loop right] node {$1/2 + \eps; s_d$} (kd);    
    \end{tikzpicture}
    \caption{Action $a_w$}
\end{subfigure}
\begin{subfigure}[b]{.32\textwidth}
\centering
    \begin{tikzpicture}[->,>=stealth',shorten >=1pt,auto,node distance=3cm, semithick]
    
    \node[state] (ku) [rect] {$k_u | 1$};
    \node[state] (kd) [rect, below of = ku, yshift = 20] {$k_d | 0$};
    
    \path   
    (ku) edge [loop left] node {$1; s_u$} (ku)
    (kd) edge [loop left] node {$1; s_u$} (kd);    
    \end{tikzpicture}
    \caption{Action $a_u$}
\end{subfigure}
\begin{subfigure}[b]{.32\textwidth}
\centering
    \begin{tikzpicture}[->,>=stealth',shorten >=1pt,auto,node distance=3cm, semithick]
    
    \node[state] (ku) [rect] {$k_u | 0$};
    \node[state] (kd) [rect, below of = ku, yshift = 20] {$k_d | 1$};
    
    \path   
    (ku) edge [loop left] node {$1; s_d$} (ku)
    (kd) edge [loop left] node {$1; s_d$} (kd);    
    \end{tikzpicture}
    \caption{Action $a_d$}
\end{subfigure}
\caption{Discontinuous value POMDP}
\label{Figure: Discontinuity}
\end{figure}

Consider an initial belief $p_1 = 1/2 \cdot \delta_{k_u} + 1/2 \cdot \delta_{k_d}$. If $\eps = 0$, the value is $1/2$, achieved for example by the constant strategy $\sigma \equiv a_u$. Note that playing action $a_w$ leads to no information since the random signal the player receives is independent of the underlying state. In contrast, if $\eps > 0$, playing action $a_w$ reveals information about the underlying state. If action $a_w$ is played sufficiently many times, the player can estimate the state by comparing the number of signals $s_d$ against $s_u$: if $s_d$ appear more than $s_u$, then it is more probable that the underlying state is $k_d$. Therefore, by playing $a_w$, the player can estimate the state with increasing probability. That is why the value for $\eps > 0$ is $1$. This proves that this POMDP is discontinuous with respect to the transition function since
	\[
    1 = \liminf_{\eps \to 0} v_{\infty}(p_1 \mid \eps) 
        >  v_{\infty}(p_1 \mid \eps = 0) = \frac{1}{2} \,.
	\] 

\subsection{Illustrative examples}
\label{Sec: Illustrative examples}

We show an example of POMDP that will be analyzed in Section \ref{Sec: Illustration} in light of our technique. This example comes in two variants differing in sophistication. 

\subsubsection{Simple version}

Let us explain the easiest version.

\begin{restatable}{example}{simple}
\label{Example: Simple POMDP}
    Consider two states ($k_u, k_d$) and two actions: \emph{up} ($a_u$) and \emph{down} ($a_d$). All transitions are possible (including loops) and they do not depend on the action. Signals inform the player when the state changes. In terms of actions and rewards, by playing $a_u$ the player obtains a reward of $1$ only if the current state is $k_u$. Similarly, by playing $a_d$ the player obtains a reward of $1$ only if the state is $k_d$. Figure \ref{Figure: Simple POMDP} is a representation of the game with specific transition probabilities.
\end{restatable}

\begin{figure}[h]
\begin{subfigure}{.5\textwidth}
\centering
    \begin{tikzpicture}[->,>=stealth',shorten >=1pt,auto,node distance=3cm, semithick]
    
    \node[state] (ku) [rect] {$k_u | 1$};
    \node[state] (kd) [rect, below right of = ku] {$k_d | 0$};
    
    \path   
    (ku) edge [bend left] node {$1/2; s_c$} (kd)
    (ku) edge [loop left] node {$1/2; s_w$} (ku)
    (kd) edge [bend left] node {$1/2; s_c$} (ku)
    (kd) edge [loop right] node {$1/2; s_w$} (kd);    
    \end{tikzpicture}
    \caption{Action $a_u$}
\end{subfigure}
\begin{subfigure}{.5\textwidth}
\centering
    \begin{tikzpicture}[->,>=stealth',shorten >=1pt,auto,node distance=3cm, semithick]
    
    \node[state] (ku) [rect] {$k_u | 0$};
    \node[state] (kd) [rect, below right of = ku] {$k_d | 1$};
    
    \path   
    (ku) edge [bend left] node {$1/2; s_c$} (kd)
    (ku) edge [loop left] node {$1/2; s_w$} (ku)
    (kd) edge [bend left] node {$1/2; s_c$} (ku)
    (kd) edge [loop right] node {$1/2; s_w$} (kd);    
    \end{tikzpicture}
    \caption{Action $a_d$}
\end{subfigure}
\caption{Simple POMDP}
\label{Figure: Simple POMDP}
\end{figure}

Consider an initial belief $p_1 = 1/4 \cdot \delta_{k_u} + 3/4 \cdot \delta_{k_d}$. During a play, the decision-maker can have two beliefs, $1/4 \cdot \delta_{k_u} + 3/4 \cdot \delta_{k_d}$ or $3/4 \cdot \delta_{k_u} + 1/4 \cdot \delta_{k_d}$, because the signals notify when there has been a change. The value of this game is $3/4$. An optimal strategy is to play action $a_d$ until getting a signal $s_c$, then playing action $a_u$ until getting a signal $s_c$, and repeat.

\subsubsection{Involved version}

Let us go to the more complex version. Now the transition between the two extremes includes more states, instead of being a direct jump.

\begin{restatable}{example}{involved}
\label{Example: Involved version}
    Consider six states and four actions: \emph{up} ($a_u$), \emph{down} ($a_d$), \emph{left} ($a_l$) and \emph{right} ($a_r$). States can be separated into two groups: \emph{extremes} ($k_u$ and $k_d$) and \emph{transitional} ($k_{l_1}, k_{r_1}, k_{l_2}$ and $k_{r_2}$). Furthermore, transitional states can be divided into two groups: \emph{left states} ($k_{l_1}$ and $k_{l_2}$) and \emph{right states} ($k_{r_1}$ and $k_{r_2}$). Transitions are from extreme states to transitional states and from transitional to extremes. More precisely, excluding loops, only the following transitions are possible: from $k_u$ to either $k_{l_1}$ or $k_{r_1}$, then from these two to $k_d$, from $k_d$ to either $k_{l_2}$ or $k_{r_2}$ and then back to $k_u$. Signals are such that the player knows: (i) the state changed to an extreme state, or (ii) the state changed to a transitional state and the new state is with higher probability a left state or a right state. In terms of actions and rewards, each action has an associated set of states in which the reward is $1$ and the rest is $0$: by playing $a_u$ the reward is $1$ only if the current state is $k_u$, playing $a_d$ rewards only state $k_d$, $a_l$ rewards states $k_{l_1}$ and $k_{l_2}$, and $a_r$ rewards states $k_{r_1}$ and $k_{r_2}$. Figure \ref{Figure: Complex POMDP} is a representation of this game with specific transition probabilities.
\end{restatable}
Consider an initial belief $p_1 = 1/4 \cdot \delta_{k_u} + 3/4 \cdot \delta_{k_d}$. The value of the game is $21/32$. An optimal strategy is given by playing action $a_d$ until getting a signal $s_l$ or $s_r$. If the decision-maker got signal $s_l$, then play action $a_l$, otherwise, play action $a_r$. Repeat action $a_l$ or $a_r$ until getting the signal $s_c$. Then, play $a_u$ until getting a signal $s_l$ or $s_r$. When this happens, play $a_l$ or $a_r$ accordingly until getting signal $s_c$. And so, repeat the cycle. 

The belief dynamic under this optimal strategy is the following. The initial belief is $p_1$, supported in the extreme states. By getting a signal $s_w$, the belief does not change. By getting signal $s_l$, the weight on $k_u$ distributes between states $k_{l_1}$ and $k_{r_1}$ in a proportion $3:1$ and the weight on $k_d$ distributes between $k_{l_2}$ and $k_{r_2}$ in the same way. By getting signal $s_r$, the distribution is similar, but the role of left states are interchanged with right states. Once the belief is in the transitional states, by playing the respective action (either $a_l$ or $a_r$), the belief does not change while receiving signal $s_w$. Upon receiving the signal $s_c$, the new belief is $3/4 \cdot \delta_{k_u} + 1/4 \cdot \delta_{k_d}$. By symmetry of the POMDP, the dynamic is then similar until getting signal $s_c$ for a second time. At that time, the belief is equal to the initial distribution, namely $1/4 \cdot \delta_{k_u} + 3/4 \cdot \delta_{k_d}$.

\begin{Remark}
For the decision-maker to have a finite-memory strategy, some quantity with finitely many options must be updated over time. A tentative idea is to compute the posterior belief, but it can take infinitely many values. In this example, using a belief partition is enough to encode an optimal strategy. In general, it is an open question if a belief partition is sufficient to achieve $\eps$-optimal strategies.
\end{Remark}

\begin{figure}[H]
\begin{subfigure}{.5\textwidth}
\centering
    \begin{tikzpicture}[->,>=stealth',shorten >=1pt,auto,node distance=3cm, semithick]
    
    \node[state] (ku) [rect] {$k_u | 1$};
    \node[state] (kl1) [rect, below left of = ku, yshift = 20] {$k_{l_1} | 0$};
    \node[state] (kr1) [rect, below right of = ku, yshift = 20] {$k_{r_1} | 0$};
    \node[state] (kd) [rect, below of = ku, yshift = -40] {$k_d | 0$};
    \node[state] (kl2) [rect, below left of = kd, yshift = 20] {$k_{l_2} | 0$};
    \node[state] (kr2) [rect, below right of = kd, yshift = 20] {$k_{r_2} | 0$};
    
    \path   
    (ku) edge [bend left = 50, above, sloped] node {$1/9; s_r$} (kl1)
    (ku) edge [bend right = 50, above, sloped] node {$3/9; s_r$} (kr1)
    (ku) edge [bend right = 60, above, sloped] node {$3/9; s_l$} (kl1)
    (ku) edge [bend left = 60, above, sloped] node {$1/9; s_l$} (kr1)
    (kd) edge [bend left = 50, above, sloped] node {$1/9; s_r$} (kl2)
    (kd) edge [bend right = 50, above, sloped] node {$3/9; s_r$} (kr2)
    (kd) edge [bend right = 60, below, sloped] node {$3/9; s_l$} (kl2)
    (kd) edge [bend left = 60, below, sloped] node {$1/9; s_l$} (kr2)
    (ku) edge [loop above] node {$1/9; s_w$} (ku)
    (kl1) edge [loop below] node {$1; s_w$} (kl1)
    (kr1) edge [loop below] node {$1; s_w$} (kr1)
    (kl2) edge [loop below] node {$1; s_w$} (kl2)
    (kr2) edge [loop below] node {$1; s_w$} (kr2)
    (kd) edge [loop above] node {$1/9; s_w$} (kd)
    ;    
    \end{tikzpicture}
    \caption{Action $a_u$}
\end{subfigure}
\begin{subfigure}{.5\textwidth}
\centering
    \begin{tikzpicture}[->,>=stealth',shorten >=1pt,auto,node distance=3cm, semithick]
    
    \node[state] (ku) [rect] {$k_u | 0$};
    \node[state] (kl1) [rect, below left of = ku, yshift = 20, xshift = 10] {$k_{l_1} | 1$};
    \node[state] (kr1) [rect, below right of = ku, yshift = 20, xshift = -10] {$k_{r_1} | 0$};
    \node[state] (kd) [rect, below of = ku, yshift = -40] {$k_d | 0$};
    \node[state] (kl2) [rect, below left of = kd, yshift = 20, xshift = 10] {$k_{l_2} | 1$};
    \node[state] (kr2) [rect, below right of = kd, yshift = 20, xshift = -10] {$k_{r_2} | 0$};
    
    \path   
    (kl1) edge [bend left = 25, below, sloped] node {$1/2; s_c$} (kd)
    (kr1) edge [bend right = 25, above, sloped] node {$\quad 1/2; s_c$} (kd)
    (kl2) edge [bend left = 75, above, sloped] node {$1/2; s_c$} (ku)
    (kr2) edge [bend right = 75, above, sloped] node {$1/2; s_c$} (ku)
    (ku) edge [loop above] node {$1; s_w$} (ku)
    (kl1) edge [loop below] node {$1/2; s_w$} (kl1)
    (kr1) edge [loop below] node {$1/2; s_w$} (kr1)
    (kl2) edge [loop below] node {$1/2; s_w$} (kl2)
    (kr2) edge [loop below] node {$1/2; s_w$} (kr2)
    (kd) edge [loop left] node {$1; s_w$} (kd)
    ;
    \end{tikzpicture}
    \caption{Action $a_l$}
\end{subfigure}
\begin{subfigure}{.5\textwidth}
\centering
    \begin{tikzpicture}[->,>=stealth',shorten >=1pt,auto,node distance=3cm, semithick]
    
    \node[state] (ku) [rect] {$k_u | 0$};
    \node[state] (kl1) [rect, below left of = ku, yshift = 20, xshift = 10] {$k_{l_1} | 0$};
    \node[state] (kr1) [rect, below right of = ku, yshift = 20, xshift = -10] {$k_{r_1} | 1$};
    \node[state] (kd) [rect, below of = ku, yshift = -40] {$k_d | 0$};
    \node[state] (kl2) [rect, below left of = kd, yshift = 20, xshift = 10] {$k_{l_2} | 0$};
    \node[state] (kr2) [rect, below right of = kd, yshift = 20, xshift = -10] {$k_{r_2} | 1$};
    
    \path   
    (kl1) edge [bend left = 25, above, sloped] node {$1/2; s_c \quad$} (kd)
    (kr1) edge [bend right = 25, below, sloped] node {$1/2; s_c$} (kd)
    (kl2) edge [bend left = 75, above, sloped] node {$1/2; s_c$} (ku)
    (kr2) edge [bend right = 75, above, sloped] node {$1/2; s_c$} (ku)
    (ku) edge [loop above] node {$1; s_w$} (ku)
    (kl1) edge [loop below] node {$1/2; s_w$} (kl1)
    (kr1) edge [loop below] node {$1/2; s_w$} (kr1)
    (kl2) edge [loop below] node {$1/2; s_w$} (kl2)
    (kr2) edge [loop below] node {$1/2; s_w$} (kr2)
    (kd) edge [loop right] node {$1; s_w$} (kd)
    ;
    \end{tikzpicture}
    \caption{Action $a_r$}
\end{subfigure}
\begin{subfigure}{.5\textwidth}
\centering
    \begin{tikzpicture}[->,>=stealth',shorten >=1pt,auto,node distance=3cm, semithick]
    
    \node[state] (ku) [rect] {$k_u | 0$};
    \node[state] (kl1) [rect, below left of = ku, yshift = 20] {$k_{l_1} | 0$};
    \node[state] (kr1) [rect, below right of = ku, yshift = 20] {$k_{r_1} | 0$};
    \node[state] (kd) [rect, below of = ku, yshift = -40] {$k_d | 1$};
    \node[state] (kl2) [rect, below left of = kd, yshift = 20] {$k_{l_2} | 0$};
    \node[state] (kr2) [rect, below right of = kd, yshift = 20] {$k_{r_2} | 0$};
    
    \path   
    (ku) edge [bend left = 50, above, sloped] node {$1/9; s_r$} (kl1)
    (ku) edge [bend right = 50, above, sloped] node {$3/9; s_r$} (kr1)
    (ku) edge [bend right = 60, above, sloped] node {$3/9; s_l$} (kl1)
    (ku) edge [bend left = 60, above, sloped] node {$1/9; s_l$} (kr1)
    (kd) edge [bend left = 50, above, sloped] node {$1/9; s_r$} (kl2)
    (kd) edge [bend right = 50, above, sloped] node {$3/9; s_r$} (kr2)
    (kd) edge [bend right = 60, below, sloped] node {$3/9; s_l$} (kl2)
    (kd) edge [bend left = 60, below, sloped] node {$1/9; s_l$} (kr2)
    (ku) edge [loop above] node {$1/9; s_w$} (ku)
    (kl1) edge [loop below] node {$1; s_w$} (kl1)
    (kr1) edge [loop below] node {$1; s_w$} (kr1)
    (kl2) edge [loop below] node {$1; s_w$} (kl2)
    (kr2) edge [loop below] node {$1; s_w$} (kr2)
    (kd) edge [loop above] node {$1/9; s_w$} (kd)
    ;    
    \end{tikzpicture}
    \caption{Action $a_d$}
\end{subfigure}
\caption{Complex POMDP}
\label{Figure: Complex POMDP}
\end{figure}

\clearpage


\section{Structure of the proof}
\label{Sec: Structure of the proof}
\label{sec: Proof of main Theorem}

In this section, we introduce two key lemmas and derive from them the proof of Theorem \ref{theo: POMDP}. We first define the history at stage $m$, which is all the information the decision-maker has at stage $m$.
\begin{Definition}[$m$-stage history]
Given a strategy $\sigma \in \Sigma$ and an initial belief $p_1$, denote the (random) history at stage $m$ by 
	\[
	H_m \defas ((A_1, S_1), (A_2, S_2), \ldots, (A_{m - 1}, S_{m - 1})) \,.
	\]
The random variable $H_m$ takes values in $\mc{H}_m = (\mc{A} \times \mc{S})^{m - 1}$.
\end{Definition}
Recall that we denote the state at stage $m$ by $K_m$, which takes values in $\mc{K}$; the signal at stage $m$ by $S_m$, which takes values in $\mc{S}$; and the action at stage $m$ by $A_m$, which takes values in $\mc{A}$. Note that the history at stage $m$ does not contain direct information about the states $K_1, \ldots, K_m$.

The belief of the player at stage $m$ plays a key role in the study of POMDPs, and we formally define it as follows.
\begin{Definition}[$m$-stage belief]
Given a strategy $\sigma \in \Sigma$ and an initial belief $p_1$, denote the belief at stage $m$ by $P_m$, which is given by, for all $k \in \mc{K}$,
    \[
    P_m(k) \defas \m{P}^{p_1}_{\sigma}(K_m = k \mid H_m) \,.
    \]
\end{Definition}

For fixed $\sigma$ and $p_1$, one can use Bayes rule to compute $P_m$. To avoid heavy notations, we omit the dependence of $P_m$ on $\sigma$ and $p_1$. For $p \in \Delta(\mc{K})$, denote the support of $p$ by $\Supp(p)$, which is the set of $k \in \mc{K}$ such that $p(k) > 0$.

The first ingredient of the proof of Theorem \ref{theo: POMDP} is the following lemma.
\begin{Lemma}
\label{Lemma: VZ16}
For any initial belief $p_1$ and $ \eps > 0$, there exists $m_\eps \geq 1$, $\sigma^{\eps} \in \Sigma$ and a (random) belief $P^* \in \Delta(\mc{K})$ (which depends on the history before stage $m_\eps$) such that:
\begin{enumerate}

    \item 
        \[
        \PP^{p_1}_{\sigma^{\eps}}(\left\|P_{m_\eps} - P^* \right\|_1 \leq \eps) \ge 1 - \eps \,.
        \]
    
    \item 
    \label{ergodic}
    There exists $\sigma \in \Sigma$, which depends on $P^*$, such that for all $k \in \Supp(P^*)$
    \begin{equation*} 
        \left( \frac{1}{n} \sum_{m=1}^n G_m \right) \xrightarrow[n \to \infty]{} \gamma^k_{\infty}(\sigma)  \quad \m{P}^k_{\sigma}-a.s. 
    \end{equation*}
    Moreover, $\gamma_{\infty}^{P^*}(\sigma) = v_{\infty}(P^*)$ and 
    $\EE^{p_1}_{\sigma^\eps}(v_{\infty}(P^*)) \geq v_{\infty}(p_1)-\eps$. 
\end{enumerate}
\end{Lemma}

This result is a consequence of Venel and Ziliotto \cite[Lemma 33]{VZ16}. This previous work states the existence of elements $\mu^* \in \Delta(\Delta(\mc{K}))$ and $\sigma^* \in \Delta(\Sigma)$ with similar properties to those of $P^* \in \Delta(\mc{K})$ and $\sigma \in \Sigma$. In this sense, the present lemma can be seen as a deterministic version of this previous result. 
To focus on the new tools we introduce in this paper to prove Theorem \ref{theo: POMDP}, we relegate the proof and the explanation of the differences between the two lemmata to Appendix \ref{App: Lemma VZ16}.

\begin{Remark}
The first property of Lemma \ref{Lemma: VZ16} follows immediately from \cite[Lemma 33]{VZ16} by the type of convergence in this previous result. On the other hand, the second property requires the introduction of a certain Markov chain on $\mc{K} \times \mc{A} \times \Delta(\mc{K})$. This Markov chain is already present in the work \cite{VZ16} but was used for other purposes. 
Therefore, the proof consists mainly of recalling previous results and constructions. 
\end{Remark}

\begin{Remark}
Note that $\m{P}^k_{\sigma}$ represents the law on plays induced by the strategy $\sigma$, conditional on the fact that the initial state is $k$. This does not mean that we consider the decision-maker to know $k$. In the same fashion, $\gamma^k_{\infty}(\sigma)$ is the reward given by the strategy $\sigma$, conditional on the fact that the initial state is $k$. Even though $\sigma$ is optimal in $\Gamma(P^*)$, this does not imply that $\sigma$ is optimal in $\Gamma(\delta_k)$: 
we may have $\gamma^k_{\infty}(\sigma)<v_{\infty}(\delta_k)$. 
\end{Remark}

The importance of Lemma \ref{Lemma: VZ16} comes from the fact that the average rewards converge almost surely to a limit that only depends on the initial state $k$. Intuitively, this result means that for any initial belief $p_1$, after a finite number of stages, we can get $\eps$-close to a belief $P^*$ such that the optimal reward from $P^*$ is, in expectation, almost the same as from $p_1$, and moreover from $P^*$ there exists an optimal strategy that induces a strong ergodic behavior on the state dynamics. Thus, there is a natural way to build a $3\eps$-optimal strategy $\tilde{\sigma}$ in $\Gamma(p_1)$: first, apply the strategy $\sigma^{\eps}$ for $m_\eps$ stages, then apply $\sigma$. Since after $m_\eps$ steps the current belief $P_{m_\eps}$ is $\eps$-close to $P^*$ with probability higher than $1-\eps$, the reward from playing $\tilde{\sigma}$ is at least the expectation of $\gamma^{P^*}_{\infty}(\sigma) - 2\eps$, which is greater than $v_{\infty}(p_1) - 3\eps$. Therefore, this procedure yields a $3 \eps$-optimal strategy. Nonetheless, $\sigma$ may not have finite memory, and thus $\tilde{\sigma}$ may not have either. The main difficulty of the proof is to transform $\sigma$ into a finite-memory strategy. We formalize this discussion below. 
\begin{Definition}[ergodic strategy]
    Let $p^* \in \Delta(\mc{K})$. We say that a strategy $\sigma$ is \emph{ergodic} for $p^*$ if the following holds for all $k \in \Supp(p^*)$
        \[
        \left( \frac{1}{n} \sum_{m=1}^n G_m \right) \xrightarrow[n \to \infty]{} \gamma^k_{\infty}(\sigma)  \quad \m{P}^k_{\sigma}-a.s. 
        \]
\end{Definition}
From the previous discussion, we aim at proving the following result.
\begin{Lemma} 
\label{Lemma: From ergodic to finite}
    Let $p^* \in \Delta(\mc{K})$ and $\sigma$ be an ergodic strategy for $p^*$. For all $\eps>0$, there exists a finite-memory strategy $\sigma'$ such that
\begin{equation*}
    \gamma^{p^*}_{\infty}(\sigma') \geq \gamma^{p^*}_{\infty}(\sigma) - \eps \,.
\end{equation*}
\end{Lemma}

This is our key lemma and the main technical contribution. The next section is devoted to explaining the technique used and proving it.

\begin{proof}[Proof of Theorem \ref{theo: POMDP} assuming Lemmas \ref{Lemma: From ergodic to finite} and \ref{Lemma: VZ16}]
Let $p_1$ be an initial belief and $\eps>0$. Let $m_\eps$, $\sigma^{\eps}$, $P^*$ and $\sigma$ be given by Lemma \ref{Lemma: VZ16}. 
Define the strategy $\sigma^0$ by: playing $\sigma^\eps$ until stage $m_\eps$, then switch to the  strategy $\sigma'$ given by Lemma \ref{Lemma: From ergodic to finite} for $\sigma$ and $p^*=P^*$. Note that $\sigma^0$ has finite memory. We have 
\begin{align*}
    \gamma^{p_1}_{\infty}(\sigma^0)
        &=\m{E}^{p_1}_{\sigma^\eps}\left(\gamma^{P_{m_\eps}}_{\infty}(\sigma')\right) 
        	&; \text{def } \sigma^0 \\
        &\geq \m{E}^{p_1}_{\sigma^\eps}\left( \gamma^{P^*}_{\infty}(\sigma') \right) - 2\eps 
        	&; \text{Lemma } \ref{Lemma: VZ16} \\
        &\geq \m{E}^{p_1}_{\sigma^\eps}\left( \gamma^{P^*}_{\infty}(\sigma) \right) - 3\eps 
        	&; \text{Lemma } \ref{Lemma: From ergodic to finite} \\
        &= \m{E}^{p_1}_{\sigma^\eps}\left( v_\infty(P^*) \right) - 3\eps 
        	&; \text{Lemma } \ref{Lemma: VZ16} \\
        &\geq v_{\infty}(p_1) - 4\eps 
        	&; \text{Lemma } \ref{Lemma: VZ16} \,,
\end{align*}
and the theorem is proved. 

\end{proof}

\section{Super-support and proof of Lemma \ref{Lemma: From ergodic to finite}}
\label{sec: Technique: Super-support}
In this entire section, fix $p^* \in \Delta(\mc{K})$, which will be used as an initial belief, and $\sigma$ an ergodic strategy for $p^*$.

\subsection{Notation}

For $a, b \in \RR$, denote the set $[a, b] \cap \ZZ$ by $[a \until b]$. 

\begin{Definition}[Value partition]
    Let $\sim$ be the equivalence relationship on $\Supp(p^*)$ defined by $k \sim k'$ if and only if $\gamma^k_\infty(\sigma) = \gamma^{k'}_\infty(\sigma)$. Let $\{\mc{K}_1, \ldots, \mc{K}_I\}$ be the corresponding \emph{value partition}. 
\end{Definition}

\begin{Definition}[Super-support]
    For $i \in [1 \until I]$ and $m \ge 0$, define, for all $h_m \in \mc{H}_m$,
        \[
        \mc{B}^i_m(h_m) \defas \bigcup_{k_1 \in \mc{K}_i} \{ k \in \mc{K} : \PP^{k_1}_\sigma( H_m = h_m ) > 0, \PP^{k_1}_\sigma(K_m = k | H_m = h_m) > 0 \} \,.
        \]
	In other words, $\mc{B}^i_m(h_m)$ is the set of all reachable states at stage $m$ starting from some state in $\mc{K}_i$, playing the strategy $\sigma$ and obtaining history $h_m$ (if $H_m = h_m$ is possible). Denote
        \[
        B^i_m \defas \mc{B}^i_m(H_m)\,,
        \]
	the random set associated with $H_m$, and $B_m  \defas  (B^1_m, \ldots, B^I_m)$ the \emph{super-support} at stage $m$.
\end{Definition}

\begin{Remark}
Note that
	\[
	\Supp(P_m) = \bigcup_{i \in [1 \until I]} B^i_m .
	\]
Therefore, the support of $P_m$ can be deduced from the super-support $B_m$. On the other hand, $B_m$ can not be deduced from $P_m$, and thus can not be deduced from the support of $P_m$. 
This justifies the vocabulary. 
\end{Remark}

We will build a finite-memory $\eps$-optimal strategy that plays by blocks. Each block has fixed finite length and, within each block, the strategy depends only on the history in the block and on the super-support at the beginning of the block. At the end of the block, the automaton computes the new super-support according to the block history and the previous super-support. Thus, the only difference with a bounded recall strategy is that our strategy keeps track of the super-support. Super-support is a type of origin information: it is related to the value partition, and therefore to where the current mass distribution comes from.

\begin{Definition}[$h_m$-shift]
    Let $m \geq 1$ and $h_m \in \mathcal{H}_m$. The \emph{$h_m$-shift} of $\sigma$ is the strategy $\sigma[h_m]$ defined by, for all $m' \geq 1$,     
        \[
        \sigma[h_m](h_{m'})  \defas  \sigma(h_m,h_{m'}) \,.
        \]
    We denote $\sigma_m \defas \sigma[H_m]$, the corresponding random shift at stage $m$.
\end{Definition}

In other words, $\sigma[h_m]$ corresponds to the continuation of the strategy $\sigma$ conditional on the fact that the history of the first $m$ stages was $h_m$. 

\subsection{Illustration}
\label{Sec: Illustration}

The super-support captures specific information related to the beginning of the game: the origin of the current mass distribution (given by $P_m$) in terms of the initial value partition $(\mc{K}_i)_{i \in [1 \until I]}$. There are finitely many possible super-supports and it is possible to keep track of the current super-support using Bayesian updating. Therefore, it is a good variable to be used in finite-memory strategies.

Let us recall our simple example of a POMDP, Example \ref{Example: Simple POMDP}.

\simple*

Finite-recall is enough to approximate the value of this POMDP: the decision-maker can recall the last action. Then, upon seeing the signal $s_c$, the player has to change actions. Recall that $p_1 = 1/4 \cdot \delta_{k_u} + 3/4 \cdot \delta_{k_d}$. Therefore, an optimal strategy is given by playing $a_d$ until getting signal $s_c$, then playing $a_u$ until getting signal $s_c$ and repeat. This strategy is ergodic for $p_1$ and the corresponding value partition is given by $( \mc{K}_1 = \{k_u\}, \mc{K}_2 = \{k_d\} )$, because, if $K_1 = k_u$, the long-run reward is $0$ and, if $K_1 = k_d$, the long-run reward is $1$. In this case, the super-support describes completely the belief $P_m$ since it keeps track of which state has the highest (or lowest) probability.

Although the example is simple, we can already see the difference between support strategies and super-support strategies. In this case, all strategies based on the current support (the support of $P_m$) are constant and therefore can achieve a long-run reward of at most $1/2$. On the other hand, super-support strategies can be optimal and achieve a long-run reward of $3/4$.

This example also shows that playing by blocks and defining the behaviour in each block by the current support (instead of the super-support) is not enough.

Let us analyze now our more complex POMDP example, Example \ref{Example: Involved version}.

\involved*

Recall that $p_1 = 1/4 \cdot \delta_{k_u} + 3/4 \cdot \delta_{k_d}$ and that an optimal strategy is given by playing action $a_d$ until getting a signal $s_l$ or $s_r$. If the decision-maker gets signal $s_l$, then play action $a_l$, otherwise, play action $a_r$. Repeat action $a_r$ until getting the signal $s_c$. Then, play $a_u$ until getting a signal $s_l$ or $s_r$. When this happens, play $a_l$ or $a_r$ accordingly until getting signal $s_c$. And so, repeat the cycle. 

This optimal strategy is ergodic for $p_1$ and the corresponding value partition is given by $( \mc{K}_1 = \{k_u\}, \mc{K}_2 = \{k_d\} )$ because, if $K_1 = k_u$, the long-run reward is $0$ and, if $K_1 = k_d$, the long-run reward is $7/8$. 
Contrary to the previous example, the super-support does not describe completely the belief $P_m$. Indeed, consider the initial belief $p_1$, which is supported on the extreme states, and that the decision-maker gets either signals $s_l$ or $s_r$. Then, the new belief is supported in all the transitional states and the super-support is the same under any of these two histories, and equal to: $B = (B^1 = \{k_{l_1}, k_{r_1}\}, B^2 = \{k_{l_2}, k_{r_2}\})$. Based on this super-support one can not reconstruct the current belief, but one knows more than only the support: we can differentiate the origin ($k_u$ or $k_d$) of the current belief distribution.

Notice that using the super-support alone is not enough to get $\eps$-optimal strategies. Indeed, in transitional states, the decision-maker needs to know whether the state is more likely to be in a left state or a right state in order to play well, and the super-support does not contain such information. That is why, in the proof of Lemma \ref{Lemma: From ergodic to finite}, we consider a more sophisticated class of strategies, that combine super-support and bounded recall. For the moment, let us describe such a strategy for this example. Choose $n_0$ very large, and for each $\ell \geq 1$, play the following strategy in the time block $[\ell n_0+1 \until (\ell+1)n_0]$:
\begin{itemize}
\item[-] \emph{Case 1: the super-support at stage $\ell n_0+1$ is $( \{k_u\}, \{k_d\} )$}. Play the previous 0-optimal strategy, that is: play action $a_d$ until getting a signal $s_l$ or $s_r$. If the decision-maker gets signal $s_{l}$, then play action $a_{l}$, otherwise, play action $a_r$. Repeat action $a_r$ until getting the signal $s_c$. Then, play $a_u$ until getting a signal $s_{l}$ or $s_r$. When this happens, play $a_{l}$ or $a_r$ accordingly until getting signal $s_c$. And so, repeat the cycle. 
\item[-] \emph{Case 2: the super-support at stage $\ell n_0+1$ is $( \{k_d\}, \{k_u\} )$}. Play the same strategy as in Case 1, except that the roles of $a_u$ and $a_d$ are switched. 
\item[-] \emph{Case 3: the super-support at stage $\ell n_0+1$ is $( \{k_{l_1}, k_{r_1}\}, \{k_{l_2}, k_{r_2}\} )$}. 
Play $a_l$ (or $a_r$) until getting the signal $s_c$. At this point, the super-support is 
$( \{k_d\}, \{k_u\} )$. Then, play as in Case 2. 
\item[-] \emph{Case 4: the super-support at stage $\ell n_0+1$ is $( \{k_{l_2}, k_{r_2}\}, \{k_{l_1}, k_{r_1}\} )$}. 
Play $a_l$ (or $a_r$) until getting the signal $s_c$. At this point, the super-support is 
$( \{k_u\}, \{k_d\} )$. Then, play as in Case 1. 
\end{itemize}
This strategy is sub-optimal during the first phase of Case 3 and Case 4, until the decision-maker receives signal $s_c$. As $n_0$ grows larger and larger, this part becomes negligible. Thus, for any $\eps>0$, there exists $n_0$ such that this strategy is $\eps$-optimal (but not optimal). 

\subsection{Properties}

Now we can state properties of super-supports when the strategy $\sigma$ is ergodic for $p^*$ and explain how rich is the structure of the random sequence of beliefs $(P_m)_{m \geq 1}$. By definition of ergodic strategies, the map $k \mapsto \gamma^k_{\infty}(\sigma)$ is constant on $\mc{K}_i$, and we denote its value by $\gamma^i_\infty$.

\begin{Lemma}[Continuation value]
    \label{Lemma: Continuation value}
    For all $m \geq 1$ and $i \in [1 \until I]$ it holds $\PP^{p^*}_\sigma$-a.s. that 
        \[
        \forall k \in B^i_m \quad 
            \gamma^k_\infty(\sigma_m) = \gamma^i_\infty \,
        \]
    Consequently, $B^1_m, \ldots, B^I_m$ are disjoint $\PP^{p^*}_\sigma$-a.s. 
\end{Lemma}

\begin{proof}
    Let $i \in [1 \until I]$. Considering the law given by $\PP^{p^*}_\sigma$, fix a realization $K_m = k \in B^i_m$. By definition of super-support, there exists $k' \in \mc{K}_i \subseteq \Supp(p^*)$ such that $k$ can be reached from $k'$ in $m$ steps. Recall that, since $\sigma$ is ergodic for $p^*$, 
        \[
        \left( \frac{1}{n} \sum_{m'=m}^{m+n-1} G_{m'} \right)
            \xrightarrow[n \to \infty]{} \gamma^{k'}_{\infty}(\sigma)=\gamma^i_{\infty} \quad \PP^{k'}_{\sigma} - a.s. \,.
        \]
    In particular, the convergence holds when $K_m = k$. Then,
        \[
        \left( \frac{1}{n} \sum_{m'=1}^{n} G_{m'} \right)
            \xrightarrow[n \to \infty]{} \gamma^{k'}_{\infty}(\sigma)=\gamma^i_{\infty} \quad \PP^{k}_{\sigma_m} - a.s. \,,
        \]
    and therefore $\gamma^k_{\infty}(\sigma_m) = \gamma^i_{\infty}$.
\end{proof}

Another property of the super-support is concerned with consecutive conditioning and is fairly intuitive. We formally state it in the following lemma and show the proof for completeness.

\begin{Lemma}[Continuation super-support]
\label{Lemma: Continuation super-support}
    Let $i \in [1 \until I]$, $m, m' \geq 0$. For all realizations $H_{m + m'} = h = (h_m, h_{m'}) \in \mc{H}_{m + m'}$, denoting $\mc{C}  = \mc{B}^i_{m + m'}(h_{m + m'})$, we have that, for all $k \in \mc{B}^i_m(h_m)$, 
        \[
        \PP^{k}_{\sigma[h_m]}(K_{m'} \in \mc{C} | H_{m'} = h_{m'}) = 1 \,.
        \]
\end{Lemma}

In other words, the super-support that arises at stage $m+m'$, $B_{m + m'}$, coincides with the super-support that would arise from a two-step procedure: first, advancing $m$ stages; and then, applying the continuation of the strategy, $\sigma_m$, for $m'$ more stages.

\begin{proof} 
Fix a realization $H_{m + m'} = h = (h_m, h_{m'})$ and let $k \in \mc{B}^i_m(h_m)$. Recall that, by definition of super-support,
    \[
     \mc{B}^i_m (h_m)
        = \bigcup_{\substack{\bar{k}_1 \in \mc{K}_i \\\PP^{\bar{k}_1}_{\sigma}(h_m > 0)}} \Supp\left( \PP^{\bar{k}_1}_{\sigma}(K_m = \cdot \mid H_m = h_m) \right) \,.
    \]
Therefore, there exists $\bar{k}_1 \in \mc{K}_i$ such that $k \in \Supp( \PP^{\bar{k}_1}_{\sigma}(K_m = \cdot \mid H_m = h_m) )$. In particular, we have that $\PP^{\bar{k}_1}_{\sigma}( K_m = k ) > 0$. 

Consider $k'$ such that $\PP^{k}_{\sigma[h_m]}( K_{m'} = k' \mid H_{m'} = h_{m'}) > 0$. By a semi-group property, we deduce that
    \[
    \PP^{\bar{k}_1}_{\sigma}( K_{m + m'} = k' \mid H_{m + m'} = h ) > 0 \,,
    \]
which implies that $k' \in \mc{B}^i_m(h) = \mc{C}$, and thus the lemma is proved. 
\end{proof}

\begin{Remark}
This property does not depend on the fact that $\sigma$ is ergodic for $p^*$.
\end{Remark}

\subsection{Proof of Lemma \ref{Lemma: From ergodic to finite}}

Fix $p^* \in \Delta(\mc{K})$ such that $\sigma$ is ergodic for $p^*$. Note that, for all $m \geq 1$, $B_m \in \{(\mc{C}_1, \ldots, \mc{C}_I) : \mc{C}_1, \ldots, \mc{C}_I \subseteq \mc{K}\}$, which is a finite set. Denote all different super-supports that can occur with positive probability by $\mc{D}^1, \mc{D}^2, \ldots, \mc{D}^J$, i.e.
	\[
	\left\{ \mc{D}^1, \mc{D}^2, \ldots, \mc{D}^J \right\} \defas \bigcup_{m \geq 1} \Supp B_m \,.
	\]
Moreover, since $\mc{D}^j$ corresponds to a super-support that occurs at some stage and under some history, there exists $h^j$ and $m_j$ such that $h^j \in \mc{H}_{m_j}$ and $\mc{D}^j = \mc{B}^i_m(h^j)$. In other words, $\mc{D}^j$ is the realization
of the super-support at stage $m_j$ under history $h^j$ and $\left\{ \mc{D}^1, \mc{D}^2, \ldots, \mc{D}^J \right\}$ contains all super-supports that can occur.

\paragraph{Definition of the strategy $\sigma'$.}
Let $\eps > 0$. By Lemma \ref{Lemma: Continuation value}, there exists $n_0 \in \m{N}^*$ such that for all $i \in [1 \until I]$, $j \in [1 \until J]$ and $k \in \mc{D}^j_i$, 
    \[
    \m{E}^k_{\sigma[h^j]} \left(\frac{1}{n_0} \sum_{m=1}^{n_0} G_m\right) \geq \gamma^i_{\infty}-\eps \,.
    \]

Define the strategy $\sigma'$ by blocks, and characterize each block by induction. For each $\ell \geq 0$, the block number $\ell$ consists in the stages $m$ such that $\ell n_0+1 \leq m \leq (\ell + 1) n_0$. We characterize the behavior in block $\ell$ by a variable $J_\ell \in [1 \until J]$ in the following way. For stage $m$ inside block $\ell$, the strategy $\sigma'$ plays according to $J_\ell$ and the history between stages $\ell n_0 + 1$ and $m$. Each block is characterized by induction because the variable $J_\ell$ is computed at stage $\ell n_0 + 1$ according to $J_{(\ell - 1)}$ and the history in the last $n_0$ stages. Thus, $\sigma'$ can be seen as mapping from $\cup_{m = 1}^{n_0} \mc{H}_m \times [1 \until J]$ to $\mc{A}$. 

Consider $\ell = 0$, i.e. the first block. The strategy $\sigma'$ is defined on the first $n_0$ stages as follows. Consider the value partition $\{\mc{K}_1, \ldots, \mc{K}_I\}$ given by $p^*$ and $\sigma$. By definition of $\mc{D}^1, \mc{D}^2, \ldots, \mc{D}^J$, there exists $j \in [1 \until J]$ such that $B_1 = \mc{D}^j$. Set $J_0 = j$, and define $\sigma'(h, J_0)  \defas  \sigma(h^{J_0},h)$ for all $h \in \mathcal{H}_m$ and $m \leq n_0$. 

Let us proceed to the induction step. Consider $\ell \geq 1$ and assume that we have defined $J_{(\ell - 1)}$ and $\sigma'$ up to stage $\ell n_0$. Denote the history between stages $(\ell - 1) n_0 + 1$ and $\ell n_0 + 1$ by $h \in \mc{H}_{n_0 + 1}$ and define $J_\ell$ such that, for all $i \in [1 \until I]$,
    \[
    \mc{D}^{J_\ell}_i = \mc{B}^i_{m_{J_{(\ell - 1)}} + n_0 + 1}(h^{J_{(\ell - 1)}}, h) \,. 
    \]
Then, extend $\sigma'$ for $n_0$ additional stages as before: $\sigma'(h, J_\ell)  \defas  \sigma(h^{J_\ell}, h)$ for all $h \in \mathcal{H}_m$ and $m \leq n_0$. Thus, we have defined $J_\ell$ and extended $\sigma'$ up to stage $(\ell + 1) n_0$. 

To summarize our construction in words, during stages $\ell n_0 + 1, \ell n_0 + 2, \ldots, (\ell + 1) n_0$, the decision-maker plays as if he was playing $\sigma$ from history $h^{J_\ell}$. Notice that the indexes $J_0, J_1, \ldots$ depend on the history, and therefore are random. Now we will connect the strategy $\sigma'$ with the super-support given by $p^*$ and $\sigma$.
\begin{Lemma}
\label{lemma: POMDP proof}
    For all $ i \in [1 \until I]$, $k \in \mc{K}_i$ and $\ell \geq 0$, we have that
        \[
        \PP^{k}_{\sigma'}(K_{\ell n_0+1} \in \mc{D}^{J_\ell}_i) = 1 \,.    
        \]
    Consequently, $\mc{D}^{J_\ell} = (\mc{D}^{J_\ell}_1, \mc{D}^{J_\ell}_2, \ldots, \mc{D}^{J_\ell}_I)$ is a partition of the support of $P_{\ell n_0 + 1}$. Moreover,
        \[
        \m{P}^{p^*}_{\sigma'}(K_{\ell n_0 + 1} \in \mc{D}^{J_\ell}_i) = p^*(\mc{K}_i) \,.
        \]
\end{Lemma}

\begin{proof} 
Fix $i \in [1 \until I]$ and $k \in \mc{K}_i$. We will prove the result by induction on $\ell \geq 0$. For $\ell = 0$, $\mc{D}^{J_0} = (\mc{K}_1, \dots, \mc{K}_I)$ and $\m{P}^{k}_{\sigma'}( K_{1} \in \mc{D}^{J_0}_i = \mc{K}_i) = 1$. Thus, the result holds. 

Assume $\ell \geq 1$. Note that $H_{(\ell - 1) n_0 + 1}$ determines the value of $J_0, \ldots, J_{(\ell - 1)}$. Therefore, $\mc{D}^{J_{(\ell-1)}}_i$ is also determined by $H_{(\ell - 1) n_0 + 1}$. By induction hypothesis, 
    \[
    \PP^{k}_{\sigma'}(K_{(\ell - 1)n_0+1} \in \mc{D}^{J_{(\ell - 1)}}_i) = 1 \,.
    \]
We must prove that, under these circumstances, $K_{\ell n_0+1} \in \mc{D}^{J_\ell}_i$. 

Indeed, index $J_{\ell - 1}$ defines strategy $\sigma'$ for stages $\ell n_0 + 1, \ell n_0 + 2, \ldots, (\ell + 1) n_0$: the decision-maker will play according to $\sigma[h^{J_{(\ell - 1)}}]$. By playing $\sigma'$ during this block, a history $h \in \mathcal{H}_{n_0 + 1}$ will be collected. Let $m \defas J_{(\ell - 1)}$ and $m' \defas n_0 + 1$. By Lemma \ref{Lemma: Continuation super-support}, we have that, starting from $K_{(\ell - 1) n_0 + 1} \in \mc{D}^{J_{(\ell - 1)}}_i$, playing $\sigma[h^{J_{(\ell - 1)}}]$ during $n_0$ stages and collecting history $h \in \mathcal{H}_{n_0 + 1}$ leads to a state that, by definition of $J_\ell$, is in $\mc{D}^{J_\ell}_i$. Therefore, 
    \[
    \PP^{k}_{\sigma'}(K_{\ell n_0 + 1} \in \mc{D}^{J_\ell}_i) 
        \geq \PP^{k}_{\sigma'}(K_{(\ell - 1) n_0 + 1} \in \mc{D}^{J_{(\ell - 1)}}_i) = 1 \,,
    \]
which proves the first result. 

Now we know that the union of $\mc{D}^{J_\ell}_1, \mc{D}^{J_\ell}_2, \ldots, \mc{D}^{J_\ell}_I$ covers the support of $\m{P}^{p^*}_{\sigma'}( K_{\ell n_0 + 1} = \cdot)$. Moreover, by Lemma \ref{Lemma: Continuation value}, $\mc{D}^{J_\ell}_1, \mc{D}^{J_\ell}_2, \ldots, \mc{D}^{J_\ell}_I$ are disjoint. Since $(\mc{K}_1,\dots,\mc{K}_I)$ partitions the support of $p^*$, we have that, for all $i \in [1 \until I]$,
\begin{align*}
    \m{P}^{p^*}_{\sigma'}(K_{\ell n_0 + 1} \in \mc{D}^{J_\ell}_i) 
        &= \sum_{i' = 1}^I  \sum_{k \in \mc{K}_{i'}} p^*(k) \m{P}^{k}_{\sigma'}(K_{\ell n_0 + 1} \in \mc{D}^{J_\ell}_i) \\ 
        &= \sum_{i' = 1}^I  \sum_{k \in \mc{K}_{i'}} p^*(k) \1_{k \in \mc{K}_i} \\
        &= p^*(\mc{K}_i) \,,
\end{align*}
which proves the second property. 
\end{proof}

To finish the proof of Lemma \ref{Lemma: From ergodic to finite}, we must show that the finite-memory strategy $\sigma'$ guarantees the reward obtained by $\sigma$ up to $\eps$. The idea is that in each block we are playing some shift of $\sigma$ for $n_0$ stages. The shift is chosen so that information about the initial belief is correctly updated, while the number $n_0$ is chosen so that the expected average reward of the whole block is close to the expected limit average reward. Then, since all blocks have the same approximation error, the average considering all blocks yields approximately $\gamma^{p^*}_{\infty}(\sigma)$. This is the intuition behind the following lemma.

\begin{Lemma} \label{POMDP block}
Let $L \in \m{N}^*$. The following inequality holds:
    \begin{equation*}
    \m{E}^{p^*}_{\sigma'} \left(\frac{1}{L n_0} \sum_{m=1}^{L n_0} G_m \right) \geq \gamma^{p^*}_{\infty}(\sigma)-\eps.
    \end{equation*}
\end{Lemma}

\begin{proof}
We have, for all $\ell \ge 0$, 
\begin{align*}
    \m{E}^{p^*}_{\sigma'} \left(\frac{1}{n_0}\sum_{m = \ell n_0 + 1}^{(\ell + 1) n_0} G_m \right)
    &= \m{E}^{p^*}_{\sigma'} \left(  \m{E}^{K_{\ell n_0 + 1}}_{\sigma[h^{J_\ell}]} \left(\frac{1}{n_0}\sum_{m=1}^{n_0} G_m \right) \right)
        &; \text{def. } \sigma'
    \\
    &= \m{E}^{p^*}_{\sigma'} \left(  \sum_{i = 1}^{I} \sum_{k \in \mc{D}^{J_\ell}_i} \PP^{p^*}_{\sigma'}(K_{\ell n_0 + 1} = k)
    \m{E}^{k}_{\sigma[h^{J_\ell}]} \left(\frac{1}{n_0}\sum_{m = 1}^{n_0} G_m \right) \right)
        &; \text{Lemma \ref{lemma: POMDP proof}}
    \\
    &\geq \m{E}^{p^*}_{\sigma'} \left(  \sum_{i=1}^{I} \sum_{k \in \mc{D}^{J_\ell}_i} \PP^{p^*}_{\sigma'}(K_{\ell n_0 + 1} = k)
    \left[\gamma^i_{\infty}(\sigma) - \eps \right] \right)
        &; \text{def. } n_0
    \\
    &= \left(  \sum_{i=1}^{I} \m{P}^{p^*}_{\sigma'} ( K_{\ell n_0 + 1} \in \mc{D}^{J_\ell}_i )
    \left[\gamma^i_{\infty}(\sigma) - \eps \right] \right)
    \\
    &= \sum_{i = 1}^I p^*(\mc{K}_i) \left[\gamma^i_{\infty}(\sigma) - \eps \right]
        &; \text{Lemma \ref{lemma: POMDP proof}}
    \\
    &= \gamma_{\infty}^{p^*}(\sigma)-\eps \,.
\end{align*}

It follows that
    \[
    \m{E}^{p^*}_{\sigma'} \left(\frac{1}{L n_0} \sum_{m=1}^{L n_0} G_m \right)
        = \frac{1}{L} \sum_{\ell = 0}^{L - 1} \m{E}^{p^*}_{\sigma'} \left(\frac{1}{n_0}\sum_{ m = \ell n_0 + 1}^{(\ell + 1) n_0} G_m \right)
        \geq \gamma^{p^*}_{\infty}(\sigma) - \eps \,.
    \]
\end{proof}
To conclude, since $\sigma'$ has finite memory, we have
\begin{equation*}
    \lim_{L \to +\infty} \m{E}^{p^*}_{\sigma'} \left(\frac{1}{L n_0} \sum_{m = 1}^{L n_0} G_m \right)
        = \m{E}^{p^*}_{\sigma'} \left(\liminf_{n \rightarrow +\infty} \frac{1}{n} \sum_{m = 1}^n G_m \right)
        = \gamma_{\infty}^{p^*}(\sigma') \,,
\end{equation*}
and the above lemma implies that $\gamma_{\infty}^{p^*}(\sigma') \geq \gamma_{\infty}^{p^*}(\sigma)-\varepsilon$, which proves Lemma \ref{Lemma: From ergodic to finite}: for each ergodic strategy $\sigma$  and $\eps > 0$, one can construct a finite-memory strategy $\sigma'$ that guarantees the reward obtained by $\sigma$ up to $\eps$.

\bibliographystyle{abbrv}
\bibliography{pomdp}

\newpage
\appendix

\section{Proof of Lemma \ref{Lemma: VZ16}} 
\label{App: Lemma VZ16}


\subsection{Notation}
Recall that Lemma \ref{Lemma: VZ16} is a consequence of \cite[Lemma 33]{VZ16}. 
Thus, we start by introducing some of the terms used in \cite{VZ16}, namely: $n$-stage game, invariant measure, occupation measure and the Kantorovich-Rubinstein distance. 

\begin{Definition}[$n$-stage game]
    Given a POMDP $\Gamma = (\mc{K}, \mc{A}, \mc{S}, q, g)$, we denote $\Gamma_n$ the \emph{$n$-stage game} with a value defined by
        \[
        v_n(p)  \defas  \sup_{\sigma \in \Sigma} \gamma^p_n(\sigma) \,,
        \]
    where $\gamma^{p}_n(\sigma)  \defas  n^{-1} \EE^p_\sigma \left( \sum_{m = 1}^{n} G_m \right)$. 
\end{Definition}

\begin{Remark}
    As the notation suggests, it was proven in \cite{VZ16} that for any finite POMDP $(v_n) \xrightarrow[n \to \infty]{} v_\infty$ uniformly. The fact that $(v_n)_{n \ge 1}$ converges was proven in \cite{RSV02}. 
\end{Remark}

The set $\Delta(\mc{K})$ is equipped with its Borelian $\sigma$-algebra $\mathcal{B}(\Delta(\mc{K}))$, and $\mathcal{C}(\Delta(\mc{K}), [0, 1])$ denotes the set of continuous functions from $\Delta(\mc{K})$ to $[0,1]$. 

\begin{Definition}[Invariant measure]
    Given a POMDP $\Gamma = (\mc{K}, \mc{A}, \mc{S}, q, g)$, $\mu \in \Delta(\Delta(\mc{K}))$ and $\sigma \colon \Delta(\mc{K}) \to \Delta(\mc{A})$ measurable, we say that $\mu$ is \emph{$\sigma$-invariant} if $\forall f \in \mathcal{C}(\Delta(\mc{K}), [0, 1])$ we have that
        \[
        \int_{\Delta(\mc{K})} \EE\left[ f( \tilde{q}[p, \sigma(p)] ) \right] \mu(dp) 
            = \int_{\Delta(\mc{K})} f(p) \mu(dp) \,,
        \]
    where $\tilde{q} \colon \Delta(\mc{K}) \times \mc{A} \rightarrow \Delta(\Delta(\mc{K}))$ is the natural transition in $\Delta(\mc{K})$ from one belief to another, given by Bayes rule.
\end{Definition}
The above definition can be intuitively understood in the following way: if the initial belief is distributed according to $\mu$, and the decision-maker plays the stationary strategy $\sigma$ at stage 1, then the belief at stage 2 is distributed according to $\mu$ too. 
\begin{Remark}
    Since $v_\infty \colon \Delta(\mc{K}) \to [0, 1]$ is a continuous function, one can replace $f$ by $v_\infty$ in the previous definition. Moreover, interpreting $\sigma$ as a (mixed) stationary strategy,
     we would have that the sequence $( \EE^\mu_\sigma[v_\infty(P_m)] )_{m \ge 1}$ is constant.
\end{Remark}

\begin{Definition}[$m$-stage occupation measure]
    Given a POMDP $\Gamma = (\mc{K}, \mc{A}, \mc{S}, q, g)$, a measure $\mu \in \Delta(\Delta(\mc{K}))$ and a strategy $\sigma$, consider the following dynamic over $\Delta(\mc{K})$. First, $P_1$ is drawn according to $\mu$. Then, $(P_{n})_{n \geq 1}$ is obtained by playing according to $\sigma$. This way, for each $m \geq 1$, we have that $\Gamma$, $\mu$ and $\sigma$ induce a probability over $\Delta(\mc{K})$: for each measurable set $\mc{A} \subseteq \Delta(\mc{K})$, we can define $\PP^{\mu}_\sigma(P_m \in \mc{A})$. Therefore, the $m$-stage belief, $P_m$, is a random belief. 
    
    We denote the \emph{$m$-stage occupation measure} $z_m[\mu, \sigma] \in \Delta(\Delta(\mc{K}))$ by the law of $P_m$ over $\Delta(\mc{K})$. Formally, $z_m[\mu, \sigma] \colon \mathcal{B}(\Delta(\mc{K})) \to [0, 1]$ is given by, for all $\mc{C} \in \mathcal{B}(\Delta(\mc{K}))$, 
        \[
        z_m[\mu, \sigma](\mc{C}) =\PP^{\mu}_\sigma(P_m \in \mc{C}) \,.
        \]
    For sake of notation, we identify $\Delta(\mc{K})$ with the extreme points of $\Delta(\Delta(\mc{K}))$.
\end{Definition}

\begin{Definition}[Kantorovich-Rubinstein distance]
    For all $z, z' \in \Delta(\Delta(\mc{K}))$, define
        \[
        d_{KR}(z, z')  \defas  \sup_{f \in \mc{E}_1} \left| \int_{\Delta(\mc{K})} f(p) z(dp) - \int_{\Delta(\mc{K})} f(p) z'(dp)  \right| \,,
        \]
    where $\mc{E}_1$ is the set of 1-Lipschitz functions from $\Delta(\mc{K})$ to $[0, 1]$.
\end{Definition}

\begin{Remark}
    The set $\Delta(\Delta(\mc{K}))$ equipped with distance $d_{KR}$ is a compact metric space.
\end{Remark}

\subsection{Proof }

Now we can state \cite[Lemma 33]{VZ16}.

\begin{Lemma}
    \label{Lemma: Lemma 33 in Xavier and Bruno}
    Consider a POMDP $\Gamma$ and let $p_1 \in \Delta(\mc{K})$. There exists $\mu^* \in \Delta(\Delta(\mc{K}))$ and a (mixed) stationary strategy $\sigma^* \colon \Delta(\mc{K}) \to \Delta(\mc{A})$ such that 
    \begin{enumerate}

        \item $\mu^*$ is $\sigma^*$-invariant.

        \item For all $\eps > 0$ and $N \geq 1$, there exists $n_\eps \geq N$ and $\sigma^\eps$ a pure strategy in $\Gamma$ such that $\sigma^\eps$ is 0-optimal in the $n_\eps$-stage game $\Gamma_{n_\eps}(p_1)$ and
            \[
            d_{KR}\left( \frac{1}{n_\eps} \sum_{m = 1}^{n_\eps} z_m[p_1, \sigma^\eps] , \mu^* \right) \leq \eps \,.
            \]
        
        \item
            \[
            \int_{\Delta(\mc{K})} g(p, \sigma^*(p)) \mu^*(dp) 
                = \int_{\Delta(\mc{K})} v_\infty(p) \mu^*(dp) 
                = v_\infty(p_1) \,.
            \]

    \end{enumerate} 
\end{Lemma}

\begin{Remark}
     Lemma \ref{Lemma: Lemma 33 in Xavier and Bruno} works with elements in $\Delta(\Delta(\mc{K}))$ and a mixed stationary strategy, while Lemma \ref{Lemma: VZ16} deals with (random) elements in $\Delta(\mc{K})$ and a pure strategy. In this sense, we would like to ``go down a level'': moving from $\mu^*$ to a random $P^*$, from $z_m$ to $P_m$ and still preserve a relationship between $v_\infty(p_1)$ and $\EE_{\sigma^\eps}^{p_1}( v_\infty(P^*))$. The ergodic property \ref{ergodic} of Lemma \ref{Lemma: VZ16} follows from the first and third property of Lemma \ref{Lemma: Lemma 33 in Xavier and Bruno}.
\end{Remark}

\begin{proof}[Proof of Lemma \ref{Lemma: VZ16}]   	 
	Consider $p_1 \in \Delta(\mc{K})$ and $\eps > 0$ fixed. Since $(v_n)_{n \geq 1}$ converges uniformly to $v_\infty$, consider $N \geq 1$ such that $\eps \geq 1/N$ and  $\forall n \geq N \quad || v_n - v_\infty ||_\infty \leq \eps$. Now, using Lemma \ref{Lemma: Lemma 33 in Xavier and Bruno}, there exists $\mu^*$ and $\sigma^*$ such that $\mu^*$ is $\sigma^*$-invariant and, considering $\eps^4$, $\exists n_\eps \geq N$ such that 
        \[
        d_{KR}\left( \frac{1}{n_\eps} \sum_{m = 1}^{n_\eps} z_m[p_1, \sigma^\eps] , \mu^* \right) \leq \eps^4 \,,
        \]
    with $\sigma^\eps \in \Gamma_{n_\eps}(p_1)$ an optimal pure strategy for the $n_\eps$-stage game starting in $p_1$. 
    
    We claim that $\exists m_\eps \leq \lceil \eps n_\eps \rceil$ such that 
    \begin{equation} \label{close_to_support}
    \PP^{p_1}_{\sigma^\eps}( P_{m_\eps} \in \Supp(\mu^*) + B(0, \eps) ) > 1 - \eps. 
    \end{equation}
    Proceeding by contradiction, assume that $\forall m \leq \lceil \eps n_\eps \rceil$, we have $\PP^{p_1}_{\sigma^\eps}( P_{m} \in \Supp(\mu^*) + B(0, \eps) ) \leq 1 - \eps$. Define the function $f \colon \Delta(\Delta(\mc{K})) \to [0, 1]$ by $f(p) = d_\infty(p, \Supp(\mu^*))$, the supremum distance from $\Supp(\mu^*)$. Clearly, $f \in \mc{E}_1$. Moreover, 
    \begin{align*}
    \eps^4 
        &\geq d_{KR}\left( \frac{1}{n_\eps} \sum_{m = 1}^{n_\eps} z_m[p_1, \sigma^\eps] , \mu^* \right) \\
        &\geq \left| \int_{\Delta(\mc{K})} f(p) \frac{1}{n_\eps} \sum_{m = 1}^{n_\eps} z_m[p_1, \sigma^\eps](dp) - \int_{\Delta(\mc{K})} f(p) \mu^*(dp)  \right| \\
        &\geq \left| \frac{1}{n_\eps} \sum_{m = 1}^{n_\eps} \int_{\Delta(\mc{K})} f(p) z_m[p_1, \sigma^\eps](dp)  \right| 
            &; f(p) = 0, p \in \Supp(\mu^*) \\
        &\geq \frac{1}{n_\eps} \sum_{m = 1}^{\lceil \eps n_\eps \rceil} \int_{\Delta(\mc{K}) \setminus (\Supp(\mu^*) + B(0, \eps))} f(p) z_m[p_1, \sigma^\eps](dp) 
            &; f, z_m[p_1, \sigma^\eps] \geq 0 \\
        &\geq \eps \frac{1}{n_\eps} \sum_{m = 1}^{\lceil \eps n_\eps \rceil} z_m[p_1, \sigma^\eps]( \Delta(\mc{K}) \setminus (\Supp(\mu^*) + B(0, \eps)) 
            &;\text{def. of }f\\
        &\geq \eps^3 
            &; \text{contradiction hypothesis,}
    \end{align*}
    which is a contradiction for $\eps < 1$. Thus, we have proven \eqref{close_to_support}.
    
    Take $P^* \in \argmin_{p \in \Supp(\mu^*)} \left\|p-P_{m_\eps}\right\|_1$.
    By equation (\ref{close_to_support}), the first property of Lemma \ref{Lemma: VZ16} is satisfied. 
    
    For the second property of Lemma \ref{Lemma: VZ16}, note that, with probability higher than $1-\eps$,
    \begin{align*}
    v_\infty(P^*)
        &\geq v_{n_\eps}(P^*) - \eps
             &; ||v_{n_\eps} - v_\infty||_\infty \leq \eps \\
        &\geq v_{n_\eps}(P_{m_\eps}) - 2\eps &; v_{n_\eps} \ \text{is 1-Lipschitz} \,. 
    \end{align*}
    
    On the other hand, taking expectation we get that
    \begin{align*}
    \EE^{p_1}_{\sigma^\eps} \left( v_{n_\eps}(P_{m_\eps}) \right) 
        &\ge \EE^{p_1}_{\sigma^\eps} \left( v_{n_\eps - m_\eps}(P_{m_\eps}) \right) - \eps
             &; m_\eps \leq \lceil \eps n_\eps \rceil \\
        &\ge v_{n_\eps}(p_1) - 2 \eps
        	&; \sigma^\eps \ \text{is} \ \text{$0$-optimal} \ \text{in} \ \Gamma_{n_\eps}(p_1) \\
        &\geq v_\infty(p_1) - 3\eps
             &; \forall n \geq N \quad ||v_n - v_\infty||_\infty \leq \eps \,.                  
    \end{align*}     

    Therefore, we conclude that
        \[
        \EE^{p_1}_{\sigma^\eps} \left( v_\infty(P^*) \right) \geq v_\infty(p_1) - 6\eps \,.
        \]
        
    To complete the proof of Lemma \ref{Lemma: VZ16}, given $P^*$, we need a (pure) strategy $\sigma$ such that
    \begin{equation} \label{prop_conv}
    \forall \ k \in \Supp(P^*) \quad
        \left( \frac{1}{n} \sum_{m=1}^n G_m \right) \xrightarrow[n \to \infty]{}  \gamma^k_{\infty}(\sigma)  \quad \m{P}^k_{\sigma}-a.s. 
    \end{equation}
    and such that 
    \begin{equation} \label{prop_value}
    \gamma_{\infty}^{P^*}(\sigma) = v_{\infty}(P^*).
    \end{equation}
    
    Consider the random process $(Y_m)_{m \geq 1}$ on $\mc{Y} \defas \mc{K} \times \mc{A} \times \Delta(\mc{K})$ defined by $Y_m  \defas  (K_m, A_m, P_m)$. We claim that, under $\sigma^*$, the process $(Y_m)_{m \ge 1}$ is a Markov chain. Indeed, given $m \geq 1$ and $(Y_1, \ldots, Y_m) \in \mc{Y}^m$, $Y_{m+1}$ is generated by the following procedure:
    \begin{enumerate}
        
        \item Draw a pair $(K_{m+1}, S_m)$ according to $q(K_m, A_m)$,
        
        \item Compute $P_{m+1}$ using Bayes rule according to $P_m$ and $S_m$,
        
        \item Draw the next action $A_{m+1}$ according to $\sigma^*(P_{m+1})$.
    \end{enumerate}
    By construction, the law of $Y_{m+1}$ depends only on $Y_m$ and therefore $(Y_m)_{m \ge 1}$ is a Markov chain.
    
    Define $\nu^* \in \Delta(\mc{Y})$ by fixing the third marginal to $\mu^*$ and for all $p \in \Delta(\mc{K})$, $\nu^*( \cdot \mid p) \in \Delta(\mc{K} \times \mc{A})$ is $p \otimes \sigma^*(p)$. We claim that $\nu^*$ is an invariant measure for $(Y_m)_{m \geq 1}$. Indeed, fixing $\sigma^*$ as the strategy for the player, if $P_1$ is drawn according to $\mu^*$, then, since $\mu^*$ is $\sigma^*$-invariant, the third marginal of $Y_m$ follows $\mu^*$, for all $m \geq 1$. Moreover, conditional on $P_m$, the random variables $K_m$ and $A_m$ are independent: the conditional distribution of $K_m$ is $P_m$ and the one of $A_m$ is $\sigma^*(P_m)$. Thus, $\nu^*$ is an invariant measure of $(Y_m)_{m \ge 1}$.
    
    The strategy $\sigma^* \colon \Delta(\mc{K}) \to \Delta(\mc{A})$ is a (stationary) mixed strategy, and we are looking for a deterministic strategy $\sigma \in \Sigma$. To derandomize this strategy, note that $\sigma^*$ starting from any $p \in \Delta(\mc{K})$ is strategically equivalent to a $p$-dependent element of $\Delta(\Sigma)$, that is, a distribution over pure (not necessarily stationary) strategies (Kuhn's theorem, see \cite{F96}). To simplify notations, we still denote this equivalent strategy $\sigma^*$, and omit its dependence in $p$.
    
    Define $f \colon \mc{Y} \to [0, 1]$ by $f(k, a, p)  \defas  g(k, a)$, a measurable function. Applying an ergodic theorem in Hernández-Lerma and Lasserre \cite[Theorem 2.5.1, page 37]{HL03}, we know that $\exists f^* $ integrable with respect to $\mu^*$ such that, for all $p \in \Supp(\mu^*)$ and $\sigma \in \Supp(\sigma^*)$,
        \[
        \left( \frac{1}{n} \sum_{m = 1}^n f(K_m, A_m, P_m) 
            = \frac{1}{n} \sum_{m = 1}^n G_m 
        \right)
            \xrightarrow[n \to \infty]{} f^*(K_1, a_1, p)
        \quad \PP^{p}_{\sigma}-a.s. \,,
        \]
    where $f^*$ satisfies that $\int_{\Delta(\mc{K})} f^*(y) \nu^*(dy) = \int_{\Delta(\mc{K})} f(y) \nu^*(dy)$. 
    
    We claim that, for all $p \in \Supp(\mu^*)$ and $\sigma \in \Supp(\sigma^*)$, we have that, for all $k \in \Supp(p)$, 
        \[
        f^*(k, a_1 , p) = \gamma^{k}_\infty(\sigma) \,,
        \]
    where $a_1$ is the first action according to $\sigma$ (formally, $a_1 = \sigma(\emptyset) = \sigma^*(p)$).

    Indeed, take $p \in \Supp(\mu^*)$, $\sigma \in \Supp(\sigma^*)$ and $k \in \Supp(p)$, then
        \[
        \gamma^k_\infty(\sigma) 
            = \EE^k_{\sigma}\left( \liminf_{n \to \infty} \frac{1}{n} \sum_{m=1}^n G_m \right) 
            = \EE^k_{\sigma}\left( f^*(K_1, a_1, p) \right) 
            = f^*(k, a_1, p) \,.
        \]
    Since $\Supp(P^*) \subseteq \{ k \in \mc{K} : \exists p \in \Supp(\mu^*) \text{ s.t. } k \in \Supp(p) \}$, property \eqref{prop_conv} is satisfied. 

	Let us now turn to property \eqref{prop_value}. 
    We claim that, $\mu^*-a.s.$ and $\sigma^*-a.s.$, 
        \[
        \gamma^{p}_\infty(\sigma) = v_\infty(p) \,.
        \]
     Indeed, note that 
    \begin{align*}
        \int_{\Delta(\mc{K})} \int_{\Sigma} \gamma^{p}_\infty(\sigma) \sigma^*(d\sigma) \mu^*(dp)
            &= \int_{\Delta(\mc{K})} f^*(y) \nu^*(dy) 
                &; \text{def. of } \nu^* \\
            &= \int_{\Delta(\mc{K})} f(y) \nu^*(dy) \\
            &= \int_{\Delta(\mc{K})} g(p, \sigma^*(p)) \mu^*(dp)
                &; \text{def. of } \nu^* \\
            &= \int_{\Delta(\mc{K})} v_\infty(p) \mu^*(dp) 
                &; Lemma~\ref{Lemma: Lemma 33 in Xavier and Bruno} \,.
    \end{align*}
    By definition of $v_\infty$, for all $\sigma \in \Sigma$, $\gamma^\cdot_\infty(\sigma) \leq v_\infty(\cdot)$. Therefore, by positivity, we can conclude that, $\mu^*-a.s.$ and $\sigma^*-a.s.$, $\gamma^{p}_\infty(\sigma) = v_\infty(p)$. 
     Since the support of $P^*$ is included in $\Supp(\mu^*)$, 
    we conclude that 
    	\[
    	\gamma^{P^*}_\infty(\sigma) = v_\infty(P^*) \,,
    	\]
	and property (\ref{prop_value}) is satisfied.
\end{proof}

\end{document}